\newcommand{\CLASSINPUTbaselinestretch}{1.8}
\newcommand{\CLASSINPUTinnersidemargin}{1in}
\newcommand{\CLASSINPUTtoptextmargin}{1in}

\documentclass[journal,12pt,onecolumn,draftclsnofoot]{IEEEtran}
\usepackage[T1]{fontenc}% optional T1 font encoding
\usepackage{xcolor}
\usepackage{dblfloatfix}
\usepackage{bbm}
\usepackage{blindtext}
\usepackage{algpseudocode}
\usepackage{amsmath}
\usepackage{amssymb}
\usepackage{amsthm}
\usepackage{color}
\usepackage{cite}
\usepackage{graphicx}
\graphicspath{{./}{Figs/}}
\usepackage{subcaption}
\usepackage{mathrsfs}
\usepackage{verbatim}
\usepackage{indentfirst}
\usepackage{url}
\usepackage{epsfig,endnotes,amsthm,changepage}
\usepackage{authblk}

\usepackage{epstopdf,lipsum,etoolbox}

\usepackage{graphicx}
\usepackage{amsfonts}
\usepackage{cite}
\usepackage{array}
\usepackage{mdwmath}
\usepackage{times}
\usepackage{epsfig}
\usepackage{latexsym}
\usepackage{epstopdf}
\usepackage{verbatim}
\usepackage{units}
\usepackage{amsthm}
\usepackage{mdwlist}
\usepackage{dblfloatfix}
\usepackage{blindtext}
\usepackage[font=small,labelfont=bf]{caption}

\usepackage{tabu}

\AppendGraphicsExtensions{.pdf}
\usepackage{epstopdf,lipsum,etoolbox}

\newtheorem{definition}{Definition}
\newtheorem{theorem}{Theorem}

\newtheorem*{proposition*}{Proposition}
\newtheorem{corollary}{Corollary}

\newtheorem{lemma}{Lemma}
\newtheorem{remark}{Remark}

\allowdisplaybreaks

\usepackage{dsfont}
\usepackage{algpseudocode} 
\usepackage{algorithm}

\ifCLASSINFOpdf
\else
\fi
\usepackage{amsmath}
\begin{document}
%
% paper title
% Titles are generally capitalized except for words such as a, an, and, as,
% at, but, by, for, in, nor, of, on, or, the, to and up, which are usually
% not capitalized unless they are the first or last word of the title.
% Linebreaks \\ can be used within to get better formatting as desired.
% Do not put math or special symbols in the title.
\title{Privacy Amplification for Federated Learning via User Sampling and Wireless Aggregation}
%
%
% author names and IEEE memberships
% note positions of commas and nonbreaking spaces ( ~ ) LaTeX will not break
% a structure at a ~ so this keeps an author's name from being broken across
% two lines.
% use \thanks{} to gain access to the first footnote area
% a separate \thanks must be used for each paragraph as LaTeX2e's \thanks
% was not built to handle multiple paragraphs
%

\author{Mohamed~Seif~Eldin~Mohamed,~
        Wei-Ting~Chang,~
        and~Ravi~Tandon% <-this % stops a space
\thanks{The authors are with the Department
of Electrical and Computer Engineering, The University of Arizona, Tucson,
AZ, 85721 USA e-mail: \{mseif, wchang\}@email.arizona.edu, tandonr@arizona.edu. Parts of this paper have been submitted to IEEE International Symposium on Information Theory (ISIT) 2021. This work has been supported in part by NSF Grants CAREER 1651492, CNS 1715947, and the 2018 Keysight Early Career Professor Award.}}% <-this % stops a space
\maketitle

%\blfootnote{This work has been supported in part by NSF Grants CAREER 1651492, CNS 1715947, and the 2018 Keysight Early Career Professor Award.}

\vspace{-1in}
\begin{abstract}
In this paper, we study the problem of federated learning over a wireless channel with user sampling, modeled by a Gaussian multiple access channel, subject to central and local  differential privacy (DP/LDP) constraints. It has been shown that the superposition nature of the wireless channel provides a dual benefit of bandwidth efficient gradient aggregation, in conjunction with strong DP guarantees for the users. Specifically, the central DP privacy leakage has been shown to scale as $\mathcal{O}(1/K^{1/2})$, where $K$ is the number of users. It has also been shown that user sampling coupled with orthogonal transmission can enhance the central DP privacy leakage with the same scaling behavior. In this work, we show that, by join incorporating both wireless aggregation and user sampling, one can obtain even stronger privacy guarantees. We propose a private wireless gradient aggregation scheme, which  relies on independently randomized participation decisions by each user. The central DP leakage of our proposed scheme scales as $\mathcal{O}(1/K^{3/4})$. In addition, we show that LDP is also boosted by user sampling. We also present analysis for the convergence rate of the proposed scheme and study the tradeoffs between wireless resources, convergence, and privacy theoretically and empirically for two scenarios when the number of sampled participants are $(a)$ known, or $(b)$ unknown at the parameter server.
\end{abstract}


\newpage
\section{Introduction}
\label{sec:introduction}
Federated learning (FL) \cite{mcmahan2017communication} is a framework that enables multiple users to jointly train a machine learning (ML) model with the help of a parameter server (PS), typically, in an iterative manner. 
In this paper, we focus on a variation of FL termed federated stochastic gradient descent (FedSGD), where users compute gradients for the ML model on their local datasets, and subsequently exchange the gradients for model updates at the PS. 
There are several motivating factors behind the surging popularity of FL: $(a)$ centralized approaches can be inefficient in terms of storage/computation, whereas FL provides natural parallelization for training, and $(b)$ local data at each user is never shared, but only the local gradients are collected. However, even exchanging gradients in a raw form can leak information, as demonstrated in recent works \cite{shokri2017membership, hayes2019logan, melis2019exploiting,triastcyn2019federated, agarwal2018cpsgd, li2018federated, chen2018stochastic}. In addition, exchanging gradients incurs significant communication overhead. Therefore, it is crucial to design training protocols that are both communication efficient and private.

% Several methods can be used to reduce the transmission cost and  improve communication efficiency. One such method is user sampling, which was instantiated in recent works \cite{goetz2019active, cho2020client, li2019convergence} for the FL problem. The idea behind user sampling for FL is to reduce the number of transmissions and/or the number of transmitting users in a particular training iteration. In \cite{goetz2019active, cho2020client, cho2020bandit}, only users who have high estimated contributions, which are estimated using local losses, transmit their gradients to the PS. Other methods used for mitigating communication overhead, include transmission of compressed gradients \cite{mcmahan2017communication, lin2017compress, Farzin2020compress}.

% Another parallel recent trend is to study the feasibility of FL over wireless channels. 
% As the prototypical computation for FL training involves gradient aggregation from multiple users, the superposition property of the wireless channel can naturally support this operation much more efficiently. However, channel resources are even more constrained in the wireless setting. Thus, the communication efficiency needs to be further improved.
Since the training of FedSGD involves gradient aggregation from multiple users, the superposition property of wireless channels can naturally support this operation. 
Several recent works \cite{amiri2019machine, amiri2019federated, timchang2020FL, zhu2018low, yang2020federated, amiri2019over, zeng2019energy, sery2019analog, wang2019adaptive, abad2019hierarchical, khan2019federated} have focused on exploiting the wireless channel to alleviate the communication overhead of FL.
Depending on the transmission strategy, wireless FL can be broadly categorized into digital or analog schemes. In digital schemes, gradients from each user are compressed and transmitted to the PS using a multi-access scheme. 
Digital schemes were proposed in \cite{amiri2019machine, amiri2019federated, timchang2020FL}, where in \cite{amiri2019machine} the gradient vectors are first sparsified and quantized locally at the users by setting the desired number of top elements in magnitude to one value before transmissions. In \cite{amiri2019federated}, the authors modify the digital scheme in \cite{amiri2019machine} to allow only the user with the best channel condition to transmit. In \cite{timchang2020FL}, the authors tailor the quantization scheme to the capacity region of the underlying MAC, which allows the gradient vectors to be quantized according to both informativeness of the gradients and the channel conditions. However, digital schemes require the PS to decode individual gradients and then aggregate them.

For analog schemes, on the other hand, gradients are rescaled at each user to satisfy the power constraint and to mitigate the effect of channel noise.
All users then transmit the rescaled gradients via wireless channel simultaneously.
Non-orthogonal over the air aggregation  makes analog schemes more bandwidth efficient compared to digital ones. There have been several recent works focusing on the design of analog schemes for wireless FL. In \cite{zhu2018low, yang2020federated}, wireless aggregation is done by aligning the gradients through power control or beamforming. The communication efficiency is further enhanced by incorporating user scheduling. In addition to power control, \cite{amiri2019machine, amiri2019federated, amiri2019over} project the gradients to lower dimension prior to transmissions to improve communication efficiency, where \cite{amiri2019over} also utilizes user scheduling and only allows users with good channel conditions to transmit. In \cite{zeng2019energy}, the authors focus on minimizing the energy consumption of users in wireless FL by formulating and solving an optimization problem subject to latency constraints. In \cite{sery2019analog}, the authors proposed a gradient-based multiple access algorithm that let users transmit analog functions using common shape waveforms to mitigate the impact of fading. In \cite{wang2019adaptive}, the authors provide convergence analysis for wireless FL with non-i.i.d. data. Based on the bound on the convergence rate, the authors of \cite{wang2019adaptive} optimize the frequency of global aggregation based on the data, model, and system dynamics.

There is a large body of recent work focusing on the design of differentially private FL. Differential privacy (DP) \cite{dwork2014algorithmic} has been adopted a \textit{de facto} standard notion for private data analysis and aggregation. Within the context of FL, the notion of local differential privacy (LDP) is more suitable in which a  user can locally perturb and disclose the data to an \textit{untrusted} data curator/aggregator \cite{joseph2018local}. 
% LDP has been already adopted and used in current applications, including Google's RAPPOR \cite{fanti2016building} for website browsing history aggregation, and by Microsoft for privately collecting telemetry data \cite{ding2017collecting}. 
In the literature, there have been several research efforts to design FL algorithms satisfying LDP \cite{ geyer2017differentially, choudhury2019differential},
% While LDP provides stronger privacy guarantees (compared to a centralized solution), this comes at the cost of lower utility. 
% In particular, to achieve the same level of privacy attained by a centralized solution, significantly higher amount of noise/perturbation  is needed \cite{cormode2018privacy, wang2018empirical, bassily2019linear, bassily2015local, bassily2017practical, smith2017interaction}. 
which require significant amount of perturbation noise to ensure privacy guarantees.
However, the amount of noise can be further reduced when employing user sampling \cite{balle2018privacy}, where users are sampled by the PS to participate in the training in each iteration.
% ; $(2)$ shuffling \cite{erlingsson2019amplification}, where one user is scheduled to participate in the training in each iteration. A trusted third party shuffles the order of each user's participation. 
However, sampling schemes can be challenging in practice since they require coordination between the PS and users, and may not be feasible if the PS is untrustworthy. 
% Shuffling scheme requires a trusted third party, which may not be available in practice. 
Hence, decentralized sampling schemes that do not depend on the PS
% or a third party 
for coordination are desirable. 
To reduce the dependency on the PS, Balle et.al. \cite{balle2020privacy} recently proposed a Random Check-in protocol. More specifically, users have the choice to decide whether or not to participate in the training process, and when to participate during the training process.

In addition to saving bandwidth and computation, it has been shown in \cite{seif2020wireless, liu2020privacy, sonee2020efficient} that wireless FL also naturally provides strong differential privacy (DP) \cite{Dwork20061} guarantees.
Specifically, in \cite{seif2020wireless}, it was shown that the superposition nature of the wireless channel provides a stronger privacy guarantee as well as faster convergence in comparison to orthogonal transmission. The privacy level is shown to scale as $\mathcal{O}(1/\sqrt{K})$, where $K$ is the number of users in the wireless FL system. 
On the other hand, it was shown in \cite{balle2018privacy} that one can obtain a similar scaling of $\mathcal{O}(1/\sqrt{K})$ for privacy leakage through user sampling. The scheme of \cite{balle2018privacy}, however, considers orthogonal transmission from the sampled users.

% Recently, the privacy aspect of the wireless FL problem has gained significant interest \cite{seif2020wireless, liu2020privacy, sonee2020efficient}. Specifically, in \cite{seif2020wireless}, it was shown that the superposition nature of the wireless channel provides a stronger privacy guarantee as well as faster convergence in comparison to orthogonal transmission. The privacy level is shown to scale as the order of $\mathcal{O}(1/\sqrt{K})$, where $K$ is the number of users in the Wireless FL system. In \cite{liu2020privacy}, the authors show that instead of asking users to inject noise, one could ask users to scale down the transmit power and utilize the channel noise of the wireless channel for providing certain level of privacy. 

One natural question to ask is \textit{whether one could provide even stronger privacy guarantees by incorporating user sampling to the private wireless FedSGD scheme. If it does provide stronger guarantee, how much additional gain can be obtained? How can we optimally utilize the wireless resources, and what are the tradeoffs between convergence of FedSGD training, wireless resources and privacy?}

\begin{table}[t]
    \centering
    {\tabulinesep=1.2mm
    \begin{tabu} {|c|c|c|}
    \hline
     Transmission scheme    &  Without sampling & With sampling \\ \hline
       Orthogonal  & $\mathcal{O}(1)$ \cite{smith2017interaction}  & $\mathcal{O}(1/\sqrt{K})$ \cite{balle2018privacy} \\  \hline
       Wireless Aggregation & $\mathcal{O}(1/\sqrt{K})$ \cite{seif2020wireless} & ${\color{blue}\mathcal{O}(1/K^{3/4})}$ (Lemma \ref{lemma:optimalSamplingProb}) \\
       \hline
    \end{tabu}}
    \vspace{-10pt}
    \caption{Comparison for central privacy under: (1) orthogonal and (2) wireless aggregation transmissions.} 
    \label{fig:privacycomparison}
    \vspace{-30pt}
\end{table}

\textbf{Main Contributions}:  In this paper, we consider the problem of FedSGD training over Gaussian multiple access channels (MACs), subject to LDP and DP constraints. We propose a wireless FedSGD scheme with user sampling, where users are sampled uniformly or based on their channel conditions. We then study analog aggregation schemes coupled with the proposed sampling schemes, in which each user transmits a linear combination of $(a)$ local gradient and $(b)$ artificial Gaussian noise. The local gradients are processed as a function of the channel gains to \textit{align} the resulting gradients at the PS, whereas the artificial noise parameters are selected to satisfy the privacy constraints.  The existing privacy analysis in \cite{balle2018privacy, balle2020privacy} for FL with user sampling cannot be applied to our problem. The key challenge is that in each training iteration, the effective noise seen at the signal received by the PS over the wireless channel is a function of a random number of sampled users, making the DP/LDP analysis non-trivial.  Using concentration inequalities, we prove that the central privacy leakage scales as $\mathcal{O}(1/K^{3/4})$ with wireless aggregation and user sampling. We also provide convergence analysis of the proposed scheme for different sampling schemes.
To the best of our knowledge, this is the first result on wireless FedSGD with LDP and DP constraints with user sampling (see Table  \ref{fig:privacycomparison} for comparison).
  
% \textit{Paper Organization:} The remainder of the paper is organized as follows. We introduce the system model and problem statement in Section \ref{sec:system_model}. We then present our main results on privacy analysis and convergence rate in Section \ref{sec:main_results}. In Section \ref{experiments}, we conduct experiments on MNIST dataset to assess the performance of our proposed transmission schemes. Finally, we conclude the paper and discuss future directions in Section \ref{conclusion}.

\textit{Notations:} Boldface uppercase letters denote matrices (e.g., $\textbf{A}$),
boldface lowercase letters are used for vectors (e.g., $\textbf{a}$), we denote scalars by non-boldface lowercase letters (e.g., $x$), and sets by capital calligraphic letters (e.g., $\mathcal{X}$). $[K] \triangleq [1, 2, \cdots, K]$ represents the set of all integers from $1$ to $K$.  The set of natural numbers, integer numbers, real numbers and complex numbers are denoted by $\mathds{N}$, $\mathds{Z}$, $\mathds{R}$ and $\mathds{C}$, respectively.

%System model

\begin{figure}[t]
    \centering
    \includegraphics[width=0.6\columnwidth]{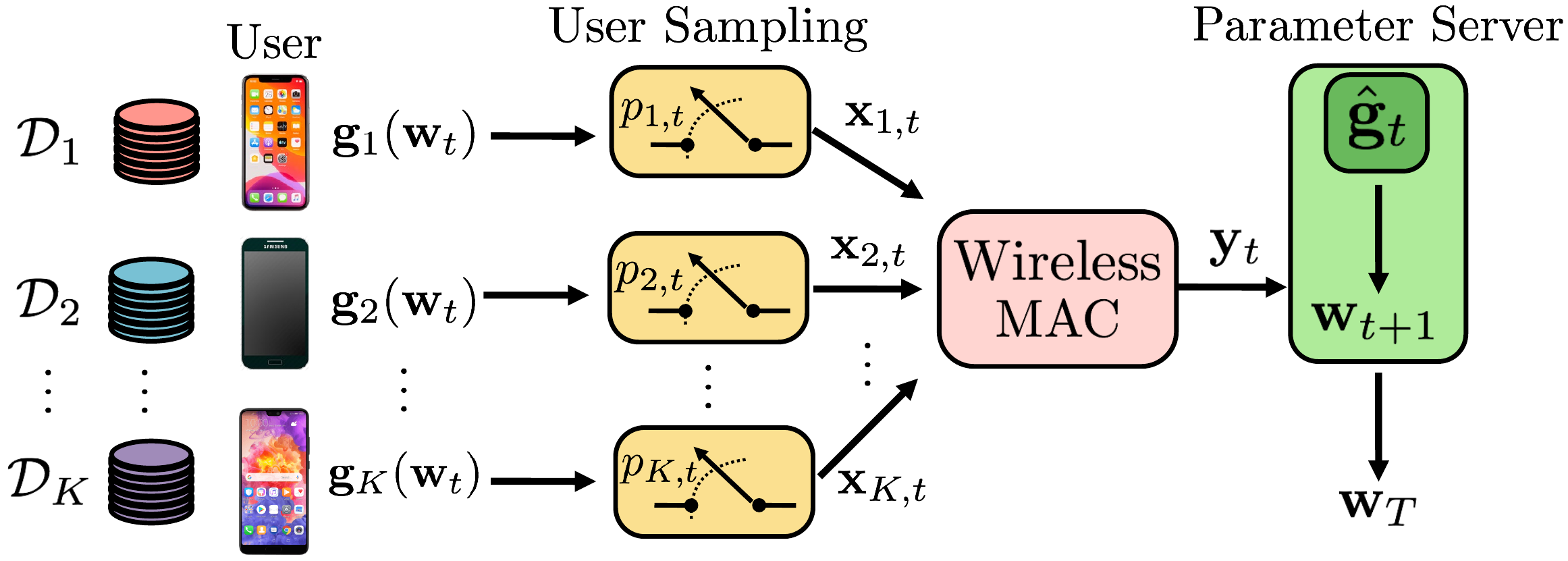}
    \vspace{-10pt}
    \caption{Illustration of the private wireless FedSGD framework: Users collaborate with the PS to jointly train a machine learning model over a Gaussian MAC. 
    % The interaction between the users and the PS must satisfy LDP constraints for each user. The final model release must satisfy certain DP level.
    }
    \label{fig:FL_system_mode}
    \vspace{-30pt}
\end{figure}

\section{System Model}\label{sec:system_model}
\textit{Wireless Channel Model}: We consider a single-antenna wireless FL system with $K$ users and a central PS. Users are connected to the PS through a Gaussian MAC as shown in Fig.   \ref{fig:FL_system_mode}. Let $\mathcal{K}_{t}$ denote the random set of users who participate in iteration $t$. The input-output relationship at the $t$-th block is

\begin{align}\label{eq:systemmodel}
      \mathbf{y}_{t} =  \sum_{k \in \mathcal{K}_{t}} h_{k,t} \mathbf{x}_{k,t}  + \mathbf{m}_{t},
\end{align}
where $\mathbf{x}_{k,t} \in \mathds{R}^{d}$  is the  signal transmitted by  user $k$ at the $t$-th block, and $\mathbf{y}_{t}$ is the received signal at the PS. Here, $h_{k,t} \geq 0$ is the channel coefficient between the $k$-th user and the PS at iteration $t$.
We assume a block flat-fading channel, where the channel coefficient remains constant within the duration of a communication block.  Each user is assumed to know its local channel gain, whereas we assume that the PS has global channel state information. Each user can transmit subject to average power constraint  i.e., $\mathds{E} \left[ \|\mathbf{x}_{k,t}\|_{2}^{2} \right] \leq P_{k}$.  $\mathbf{m}_{t}\in \mathds{R}^{d}$ is the channel noise whose elements are independent and identically distributed (i.i.d.) according to Gaussian distribution $\mathcal{N}(0, N_{0})$. 
The set of participants $\mathcal{K}_t$ can be obtained through various strategies. In this paper, we focus on user sampling, where user $k$ participates in the training at time $t$ according to probability $p_{k, t}$, for $k=1,\dots,K$. When $\mathcal{K}_t=[K]$, we recover the conventional FedSGD where every user participates in the training.

For this work, we consider $(a)$ time-invariant uniform sampling, where the sampling probability remains the same across users and iterations; $(b)$ time-variant uniform sampling, where the sampling probability remains the same across users but varies across iterations; and $(c)$ channel aware sampling, where sampling probabilities for each user can depend on the local channel gain between the user and the PS. We note that sampling strategies based on gradients or losses can potentially leak information about local datasets, hence, require analysis for privacy. Thus, we leave gradient-based sampling strategies to future work. 

 \textit{Federated Learning Problem}:  Each user $k$ has a private local dataset $\mathcal{D}_{k}$ with ${{D}_{k}} $ data points, denoted as $\mathcal{D}_{k} = \{(\mathbf{u}_{i}^{(k)}, v_{i}^{(k)})\}_{i=1}^{{{D}_{k}}}$, where $\mathbf{u}_{i}^{(k)}$ is the $i$-th data point and $v_{i}^{(k)}$ is the corresponding label at user $k$. The local loss function at user $k$ is given by
\begin{align}
    f_{k}(\mathbf{w}) = \frac{1}{{{D}_{k}}} \sum_{i=1}^{{{D}_{k}}} f(\mathbf{w}; \mathbf{u}_{i}^{(k)}, v_{i}^{(k)}) + \Omega R(\mathbf{w}),
\end{align}
where $\mathbf{w} \in \mathds{R}^{d}$ is the parameter vector to be optimized, $R(\mathbf{w})$ is a regularization function and $\Omega \geq 0$ is a regularization hyperparameter.   Users communicate with the PS through the Gaussian MAC described above in order to  train a model by minimizing the loss function $F(\mathbf{w})$, i.e., 
\begin{align}
    \mathbf{w}^{*} = \text{arg} \min_{\mathbf{w}} F(\mathbf{w}) & \triangleq \frac{1}{\sum_{k \in \mathcal{K}_{t}} D_{k} } \sum_{k \in \mathcal{K}_{t}} D_{k} {f_{k}(\mathbf{w})}.  
\end{align}
 The minimization of $F(\mathbf{w})$ is carried out iteratively through a distributed stochastic gradient descent (SGD) algorithm. More specifically, in the $t$-th training iteration, the PS broadcasts the global parameter vector $\mathbf{w}_{t}$ to all users. Each user $k$ computes his local gradient using stochastic mini batch $\mathcal{B}_k\subseteq \mathcal{D}_k$, with size $b_k$ (i.e., $|\mathcal{B}_{k}| = b_{k}$), i.e.,
\begin{align}
 \mathbf{g}_{k}(\mathbf{w}_t) = \frac{1}{{{b}_{k}}} \sum_{i\in\mathcal{B}_k}  \nabla f_{k}(\mathbf{w}_t; (\mathbf{u}_{i}^{(k)}, v_{i}^{(k)}) )   + \Omega \nabla R(\mathbf{w}_t). \label{eq:localGradSystemModel}
\end{align}
The participants, i.e., $k\in\mathcal{K}_t$, next pre-process their $\mathbf{g}_{k}(\mathbf{w}_t)$ and obtains $\mathbf{x}_{k,t}$, as explained below.  Then, the participants send their $\mathbf{x}_{k,t}$'s to the PS, where the PS receives $\mathbf{y}_t$ as defined in \eqref{eq:systemmodel}. Upon receiving $\mathbf{y}_t$, the PS performs post-processing on $\mathbf{y}_t$ to obtain $\hat{\mathbf{g}}_{t}$, the estimate of the true gradient $\mathbf{g}_{t}$ which is defined as,
\begin{align}
    {\mathbf{g}}_{t} = \frac{1}{\sum_{k =1}^{K} D_{k}}  \sum_{k =1}^{K} D_{k} \mathbf{g}_{k}(\mathbf{w}_{t}). \label{eq:exactG}
\end{align}
The global parameter $\mathbf{w}_{t}$ is updated using the estimated gradient $\hat{\mathbf{g}}_{t}$ according to $\mathbf{w}_{t+1} = \mathbf{w}_{t} - \eta_{t} \hat{\mathbf{g}}_{t}$, where $\eta_{t}$ is the learning rate of the distributed GD algorithm at iteration $t$. The iteration process continues until convergence.

Typically, in the wireless setting, the post-processing done at the PS involves removing channel effects, averaging the aggregated local gradients, and/or multiplying a constant to maintain the unbiasedness. These post-processing steps depend on the PS's knowledge of the channel condition, number of participants, and knowing how users are selected to participate. As mentioned above, the PS has global CSI. In addition, we assume that the PS knows the sampling probabilities $p_{k,t},~\forall k,t$. However, the number of participants may or may not be known at the PS. Thus, in this work, we study both cases, where $(a)$ $\mathcal{K}_t$ is known, and $(b)$ $\mathcal{K}_t$ is unknown, at the PS.

\textit{Wireless FL with User Sampling:}
The training continues for a total of $T$ iterations. Here, we describe the per-iteration operation of the algorithm. At the beginning of each iteration $t$, the PS transmits the model $\mathbf{w}_t$ to the users, and each user computes the local gradient using its local dataset according to \eqref{eq:localGradSystemModel}.
Each user $k$ participates in the training with probability $p_{k, t}$. Users then transmit their local gradients with $d$ channel uses of the wireless channel described in \eqref{eq:systemmodel}. The transmitted signal of user $k$ at iteration $t$ is given as:
\begin{align}
    \mathbf{x}_{k,t} = \begin{cases}  \alpha_{k,t} \left( {\mathbf{g}_{k}(\mathbf{w}_{t}) } + \mathbf{n}_{k,t} \right), & \text{w.p.}~p_{k, t}\\
    \mathbf{0}, & \text{otherwise}
    \end{cases}
    \label{eq:inputsignal}
\end{align}
where $\mathbf{n}_{k,t} \sim \mathcal{N}(0, {\sigma_{k,t}^2} \mathbf{I}_{d})  $ is the artificial noise term to ensure privacy, and $\alpha_{k,t}$ is the scaling  factor satisfying power constraint at each user. If a user is not participating, it does not transmit anything. We assume that the gradient vectors have a bounded norm, i.e., $\|\mathbf{g}_{k}(\mathbf{w}_{t})\|_{2} \leq L, \forall k$, and normalize the gradient vector by $L$.  The parameters $\alpha_{k,t}$s and $\sigma_{k,t}$s are designed such that the power constraints are satisfied, i.e., $\mathds{E} \left[ \|\mathbf{x}_{k,t}\|_{2}^{2} \right] = \alpha_{k,t}^{2} \left[\|\mathbf{g}_{k}(\mathbf{w}_{t})\|^{2} + d \sigma_{k,t}^{2} \right] \leq P_{k}$. From \eqref{eq:systemmodel} and \eqref{eq:inputsignal}, the received signal at the PS can be written as:
\begin{align}
    \mathbf{y}_{t} %& = \sum_{k \in \mathcal{K}_{t}} h_{k,t} \mathbf{x}_{k,t} + \mathbf{m}_{t} \nonumber \\ 
    & = \sum_{k \in \mathcal{K}_{t}}  h_{k,t} \alpha_{k,t}   {\mathbf{g}_{k}(\mathbf{w}_{t}) }  +  \underbrace{\sum_{k \in \mathcal{K}_{t}} h_{k,t} \alpha_{k,t} \mathbf{n}_{k,t} +  \mathbf{m}_{t}}_{\mathbf{z}_{t}},\label{eq:output}
\end{align}
where $\mathbf{z}_{t} \sim \mathcal{N}(0, \sigma_{{z}_{t}}^{2} \mathbf{I}_{d}) $ is the effective noise, and $\sigma_{{z}_{t}}^{2}  = \sum_{k \in \mathcal{K}_{t}} h_{k,t}^{2} \alpha_{k,t}^{2}  \sigma_{k,t}^{2} + N_{0}$. In order to carry out the summation of the local gradients over-the-air, all users pick the coefficients $\alpha_{k,t}$s in order to align their transmitted local gradient estimates. Specifically, user $k$ picks $\alpha_{k, t}$ so that 
\begin{align}
       h_{k, t} \alpha_{k,t} = 1, \forall k \in \mathcal{K}_{t}.\label{eq:alpha11}
\end{align}
For the alignment scheme described above, the received signal at the PS at iteration $t$ in \eqref{eq:output} simplifies to $\mathbf{y}_{t} = \sum_{k \in \mathcal{K}_{t}}    { \mathbf{g}_{k}(\mathbf{w}_{t}) } + \mathbf{z}_{t}$. The PS can perform two different post-processing operations to get unbiased gradient estimate $ \hat{\mathbf{g}}_t $, i.e., $\mathds{E} \left[\hat{\mathbf{g}}_{t} \right]=\mathbf{g}_t$ (see Appendix \ref{appendix:proof_theorem_2}), based on the knowledge it has: $(a)$ when $\mathcal{K}_t$ is known at the PS; $(b)$ when $\mathcal{K}_t$ is unknown at the PS.

\noindent \textbf{Case $(a)$:} When $\mathcal{K}_t$ is known at the PS, it obtains the unbiased gradient estimate $\hat{\mathbf{g}}_t$ as follows,
\begin{align}
    \hat{\mathbf{g}}_t &=  \frac{1}{\zeta_{t} |\mathcal{K}_t|}  \mathbf{y}_t = \frac{1}{ \zeta_{t}|\mathcal{K}_t|}  \sum_{k \in \mathcal{K}_{t}}  {\mathbf{g}_{k}(\mathbf{w}_{t})} +  \frac{1}{ \zeta_{t} |\mathcal{K}_t|}    \left[ \sum_{k \in \mathcal{K}_{t}}  \mathbf{n}_{k,t} +  \mathbf{m}_{t} \right], \label{eq:postProcessingKnown}
\end{align}
where $\zeta_{t} = 1 - \prod_{k=1}^{K} (1 - p_{k,t})$.

\noindent \textbf{Case $(b)$:} When $\mathcal{K}_t$ is unknown at the PS, it obtains the unbiased gradient estimate $\hat{\mathbf{g}}_t$ as follows,
\begin{align}
    &\hat{\mathbf{g}}_{t}  = \frac{1}{  \mu_{|\mathcal{K}_{t}|} }  \mathbf{y}_{t} = {  \frac{1}{\mu_{|\mathcal{K}_{t}|} }  \sum_{k \in \mathcal{K}_{t}}  {\mathbf{g}_{k}(\mathbf{w}_{t})}} + \frac{1}{ \mu_{|\mathcal{K}_{t}|} }  \left[ \sum_{k \in \mathcal{K}_{t}}  \mathbf{n}_{k,t} +  \mathbf{m}_{t} \right],
     \label{eq:postpreprocessing}
\end{align}
where $\mu_{|\mathcal{K}_{t}|}= \mathds{E} \left[|\mathcal{K}_t| \right] = \sum_{k=1}^{K} p_{k, t}$ is the expected number of participants in iteration $t$. The PS then update the models and repeats this process for $T$ iterations. 

\textit{Privacy Definitions:}
We assume that the PS is honest but curious. It is honest in the sense that it follows the FL procedure faithfully, but it might be interested in learning sensitive information about users. Therefore, the SGD algorithm for wireless FL should satisfy LDP constraints for each user. At the end of the training process, the PS may release the trained model to a third party. Thus, the training algorithm should provide central DP guarantees against any further post-processing or inference. The local and central DP are formally defined as follows:

\begin{definition} 
\vspace{-10pt}
($(\epsilon_{\ell}^{(k)}, \delta_{\ell})$-LDP \cite{{erlingsson2020encode}}) Let $\mathcal{X}_k$ be a set of all possible data points at user $k$. For user $k$, a randomized mechanism $\mathcal{M}_k: \mathcal{X}_{k} \rightarrow \mathds{R}^{d}$ is $(\epsilon_{\ell}^{(k)}, \delta_{\ell})$-LDP if for any $x,~x' \in \mathcal{X}_k$, and any measurable subset $\mathcal{O}_k \subseteq \text{Range}(\mathcal{M}_k)$, we have
\begin{align}
    \operatorname{Pr}(\mathcal{M}_k(x) \in \mathcal{O}_k) \leq \exp{(\epsilon_{\ell}^{(k)})}  \operatorname{Pr}(\mathcal{M}_k(x') \in \mathcal{O}_k) + \delta_{\ell}.
\end{align}
The setting when $\delta_{\ell} = 0$ is referred as pure $\epsilon_{\ell}^{(k)}$-LDP.
\end{definition}

\begin{definition} 
\vspace{-10pt}
($(\epsilon_{c}, \delta_{c})$-DP \cite{{erlingsson2020encode}}) Let $\mathcal{D}\triangleq \mathcal{X}_1 \times\mathcal{X}_2\times\dots\times\mathcal{X}_K$ be the collection of all possible datasets of all $K$ users. A randomized mechanism $\mathcal{M}: \mathcal{D} \rightarrow \mathds{R}^{d}$ is $(\epsilon_{c}, \delta_{c})$-DP if for any two neighboring datasets $D, D'$ and any measurable subset $\mathcal{O} \subseteq \text{Range}(\mathcal{M})$, we have
\begin{align}
    \operatorname{Pr}(\mathcal{M}(D) \in \mathcal{O}) &\leq \exp{(\epsilon_{c})}  \operatorname{Pr}(\mathcal{M}(D')  \in \mathcal{O}) + \delta_{c}. 
\end{align}
We refer to a pair of datasets $D, D' \in \mathcal{D}$ if $D'$ can be obtained from $D$ by removing one data element $x_{i}$ for some $i \in [K]$. The setting when $\delta_{c} = 0$ is referred as pure $\epsilon_{c}$-DP.
\end{definition}

%Main Results:

\section{Main Results \& Discussions\label{sec:main_results}}
% The main goal of this paper is to explore the joint benefits of wireless gradient aggregation and user sampling for privacy in FedSGD. More specifically, we study the impact of sampling probability and number of participants on the central and local privacy levels, and the convergence rate of the algorithm. 
% We also investigate the convergence rate under various assumptions, such as known and unknown $\mathcal{K}_t$.
% In this section, we present a general gradient aggregation and sampling scheme for wireless FedSGD, where each user computes the local gradient and decides whether to transmits it or not with certain probability. We then specialize this scheme in which the part of transmission containing gradients are designed in a manner so that this component is aligned at the PS. We 
% analyze this scheme and obtain the privacy leakage under LDP for each user, as a function of the wireless channel conditions, and the transmission parameters. We also analyze the synergy of user sampling and wireless aggregation on the achieved central privacy. Finally, we present the convergence rate of the private wireless FedSGD algorithm, and maximize the convergence rate by optimizing the sampling probability of each user.

\subsection{Privacy Analysis for wireless FedSGD with User Sampling}

In this section, we first derive the central DP leakage for wireless FedSGD with user sampling. Specifically, we consider two sampling strategies: $(a)$ non-uniform sampling; and $(b)$ both time-variant and time-invariant uniform sampling. For non-uniform sampling, each user can be sampled according to a probability that depends on the channel conditions. We then study a special case, i.e., uniform sampling, to understand the asymptotic behavior of the central privacy as a function of the total number of users. 
In addition, we show that user sampling is also beneficial for the local privacy level. We also quantify the gain for the local privacy level achieved by user sampling and wireless aggregation where Gaussian mechanism is used at each sampled user.
We note that the knowledge of $\mathcal{K}_t$ at the PS does not play a role in the proofs of the privacy guarantees due to the robustness of post-processing of DP.
The privacy guarantee of the proposed wireless FedSGD with non-uniform sampling is stated in the following Theorem.

\begin{theorem}\label{theorem1_central_nonUniform} 
\vspace{-10pt}
(Non-uniform sampling)  Suppose each user $k$ participates in the training process at iteration $t$ according to probability $p_{k, t}$, and utilizes local mechanism that satisfies $(\epsilon_{\ell,t}^{(k)}, \delta_{\ell})$-LDP if they decided to participate. The central privacy level of the wireless FedSGD with user sampling at iteration $t$ is given as
\begin{align}
    \epsilon_{c, t} &\leq \log \left[ 1 + \frac{\max_k p_{k, t}}{1-\delta'} \left( e^{\frac{c}{\sqrt{\mu_{|\mathcal{K}_{t}|} - \beta K}}} -1 \right)  \right], \quad \delta_{c, t} = \delta'+ \frac{\max_k p_{k, t} \delta_{\ell}}{1-\delta'},\label{eq:centralNonUni}
\end{align}
for any $\delta'\in (2e^{-2\mu_{|\mathcal{K}_{t}|}^2/K}, 1)$ and $\beta = \frac{1}{\sqrt{K}} \sqrt{0.5 \log \left(2/\delta'\right)}$, where $\mu_{|\mathcal{K}_t|}=\sum_{k=1}^{K}p_{k, t}$ denotes the expected number of users participating in iteration $t$, and $c= \frac{2L}{\sigma_{\min}} \sqrt{2\log (1.25/\delta_{\ell})}$, where $\sigma_{\min}=\min_{k, t} \sigma_{k, t}$ and $L$ is the Lipschitz constant for the loss function.
\vspace{-10pt}
\end{theorem}
The proof of the Theorem can be found in Appendix \ref{appendix:central_privacy_analysis}. The privacy parameters in \eqref{eq:centralNonUni} indicates that the central privacy leakage depends on the user with the highest sampling probability. Intuitively, a user with high sampling probability participates in the training process more often than other users with lower probabilities, thereby having most impact on the central privacy leakage. For the case with uniform sampling probability, the privacy parameters can be simplified to the following (the proof of Corollary follows directly from Theorem \ref{theorem1_central_nonUniform}):
\begin{corollary}\label{cor::uniform_sampling} 
\vspace{-10pt}
(Uniform sampling) Suppose each user decides to participate with probability $p_{k,t} = p_{t}$, and the local mechanism satisfies $(\epsilon_{\ell,t}^{(k)}, \delta_{\ell})$-LDP for each user $k$. The central privacy level of the wireless FedSGD with user sampling is given as 
\begin{align}
    \epsilon_{c, t} & \leq  \log \left[ 1 + \frac{p_t}{1-\delta'} \big(e^{\frac{c}{\sqrt{K (p_t - \beta)}}} - 1 \big) \right], \quad \delta_{c, t}  = \delta' + \frac{p_t \delta_{\ell}}{1 - \delta'}, \label{eq:uniform_central_privacy}
\end{align}
for any $\delta'\in (2e^{-2p^2K}, 1)$ and $\beta = \frac{1}{\sqrt{K}} \sqrt{0.5 \log \left(2/\delta'\right)}$, where $c= \frac{2L}{\sigma_{\min}} \sqrt{2\log (1.25/\delta_{\ell})}$.
\vspace{-10pt}
\end{corollary}

We note that both \eqref{eq:centralNonUni} (respectively, \eqref{eq:uniform_central_privacy}) is a convex function of $\{p_{k,t}\}_{k=1}^K$ (respectively, $p_t$) when $\epsilon_{\ell,t}^{(k)} \leq 1$. If the primary goal is to have strong privacy guarantee and does not need fast convergence, one can solve for the optimal sampling probabilities using the expressions in \eqref{eq:centralNonUni} and \eqref{eq:uniform_central_privacy}. However, it is difficult to obtain a closed form solution of the optimal sampling probability for the non-uniform case. Due to convexity, one can still solve it numerically using convex solvers. In contrast to the non-uniform case, one can solve for the optimal sampling probability for the uniform case as stated in the following Lemma.

\begin{lemma} 
\vspace{-10pt}
The optimal sampling probability that minimizes \eqref{eq:uniform_central_privacy} is given by 
\begin{align}
    p_t^{*} & = \min \left[ 1, \frac{2}{\sqrt{K}} \sqrt{\frac{1}{2} \log \big(\frac{2}{\delta'}\big)} \right]. \label{eqn:central_optimal_sampling}
\end{align}
By plugging $p_t^*$ back into \eqref{eq:uniform_central_privacy}, one can obtain the following upper bound on the central DP,
\begin{align}
 \epsilon_{c} &= \log \left[ \frac{2\sqrt{\frac{1}{2} \log \left(\frac{2}{\delta'}\right)}}{\sqrt{K}(1-\delta')} \big(e^{\frac{c}{\sqrt[4]{K\frac{1}{2} \log \left(\frac{2}{\delta'}\right)}}} -1 \big) + 1\right] = \mathcal{O}\left(\frac{1}{K^{3/4}}\right).
\end{align}
\label{lemma:optimalSamplingProb}
 \vspace{-30pt}
\end{lemma}
The proof of Lemma \ref{lemma:optimalSamplingProb} is presented in Appendix \ref{optimal_sampling_probability}. From Lemma \ref{lemma:optimalSamplingProb}, we observe that the central privacy level behaves as $\mathcal{O}(1/K^{3/4})$ as opposed to the $\mathcal{O}(1/\sqrt{K})$ for wireless FL without sampling \cite{seif2020wireless} and $\mathcal{O}(1/\sqrt{K})$ for FL with orthogonal transmission and user sampling \cite{balle2020privacy} (see Table \ref{fig:privacycomparison}). Clearly, when both wireless aggregation and user sampling are employed, we can obtain additional benefit in terms of central privacy. We also plot the central privacy level of the proposed scheme against other variations (see Fig. \ref{fig:OMA_vs_NOMA_sampling}).

\begin{figure}[t]
\centering
	
	\begin{minipage}[t]{.45\textwidth}
	\centering
		\subcaptionbox{\label{fig:OMA_vs_NOMA_sampling}}
	{\includegraphics[width=\linewidth]{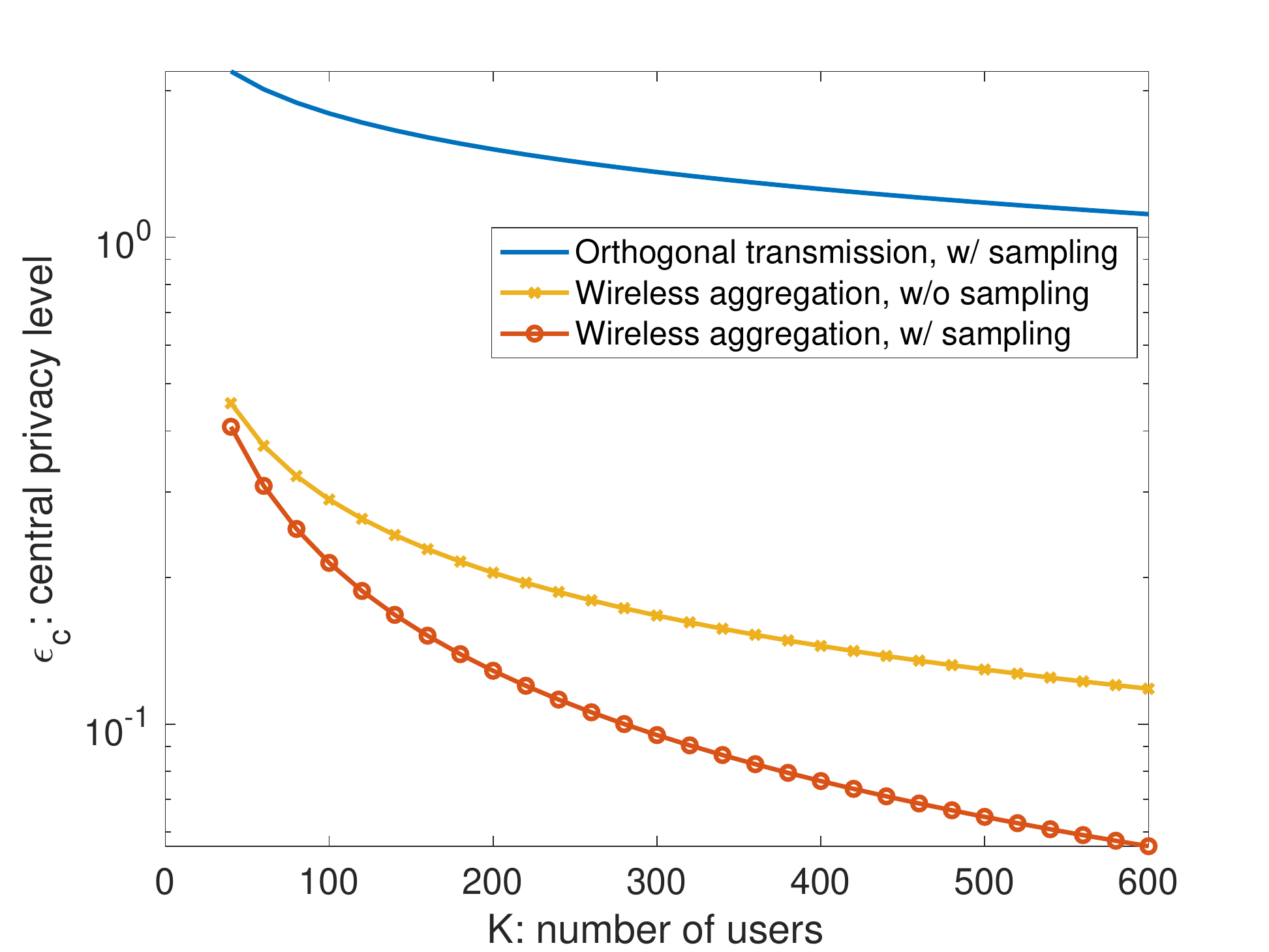}
	%\vspace{-1pt}
% 	\caption{Comparison for central privacy, where wireless aggregation with sampling is shown to outperform or--thogonal transmission with sampling and wireless aggregation without sampling, where $L=1$, $\sigma_{k,t} = N_{0} = 3$, $\delta_{l} = \delta' = 10^{-4}$ and $ p  =  \frac{2}{\sqrt{K}} \sqrt{\frac{1}{2} \log \left(\frac{2}{\delta'}\right)}$. }
	   }
	   
	%\vspace{-15pt}
	\end{minipage}
	\hspace{5pt}
	\begin{minipage}[t]{.41\textwidth}
	\centering
 \subcaptionbox{\label{fig:privacyscaling}}
	{\includegraphics[width=\linewidth]{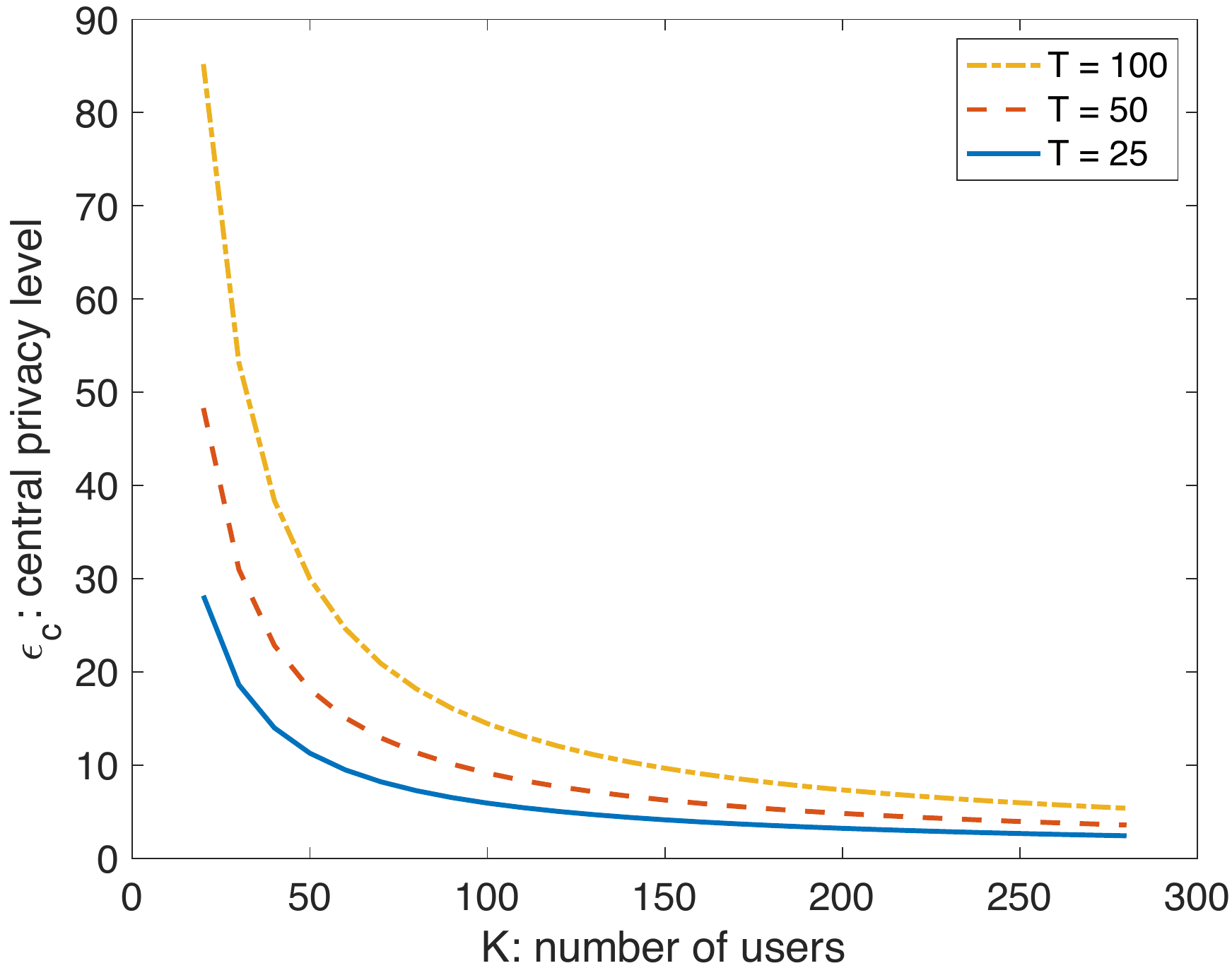}
	%\vspace{-1pt}
% 	\caption{Total privacy leakage as a function of $K$, number of users for different values of $T$, the number of training iterations, where $L=1$, $\sigma_{k,t} =N_{0} = 3$, $\delta_{l} = \delta' = 10^{-4}$ and $ p  =  \frac{2}{\sqrt{K}} \sqrt{\frac{1}{2} \log \left(\frac{2}{\delta'}\right)}$. }
 }
 \end{minipage}
 \vspace{-15pt}
    \caption{$(a)$ Comparison for central privacy, where wireless aggregation with sampling is shown to outperform other variants; $(b)$ Total privacy leakage as a function of $K$, number of users for different values of $T$, the number of training iterations, where $L=1$, $\sigma_{k,t} =N_{0} = 3$, $\delta_{l} = \delta' = 10^{-4}$ and $ p  =  \frac{2}{\sqrt{K}} \sqrt{\frac{1}{2} \log \left(\frac{2}{\delta'}\right)}$ for both figures.}
	%\vspace{-15pt}
	
% 	  \caption{Comparison for central privacy, where wireless aggregation with sampling is shown to outperform orthogonal transmission with sampling and wireless aggregation without sampling, where $L=1$, $\sigma_{k,t} = N_{0} = 3$, $\delta_{l} = \delta' = 10^{-4}$ and $ p  =  \frac{2}{\sqrt{K}} \sqrt{\frac{1}{2} \log \left(\frac{2}{\delta'}\right)}$. }
	   %\label{fig:OMA_vs_NOMA_sampling}
	   \vspace{-30pt}
\end{figure}

Interestingly, the addition of user sampling in wireless FedSGD also provides benefit for LDP. We next analyze the local privacy level achieved by the FedSGD transmission scheme.

\begin{lemma} \label{lemma1:local_privacy}
\vspace{-10pt}
For each user $k$, the proposed transmission scheme achieves $(\epsilon_{\ell,t}^{(k)}, p_{k,t} (\delta_{\ell} + \delta'))$-LDP  per iteration, where
\begin{align}
       \epsilon_{\ell,t}^{(k)} 
        % & \leq   \frac{ 2  L}{\sqrt{  \sigma_{k,t}^{2} + \kappa_t \sigma_{\min, t}^{2} + N_{0}}}  \sqrt{2 \log \frac{1.25}{\delta_{\ell}} } \nonumber \\ 
        % &\leq \frac{ 2  L}{\sqrt{ (1 + \kappa_t) \sigma_{\min, t}^{2} + N_{0}}}  \sqrt{2 \log \frac{1.25}{\delta_{\ell}} } \nonumber \\ 
        & \leq \frac{ 1}{\sqrt{  1 + \kappa_t }}  \times \frac{ 2  L}{{  \sigma_{\min,t}}}  \sqrt{2 \log \frac{1.25}{\delta_{\ell}} },
\end{align}
where $\sigma_{\min, t} \triangleq \min_{k} \sigma_{k,t}$, $\kappa_t \triangleq \sum_{i=1, i\neq k}^{K}  p_{i,t} - \beta K$, where $\beta$ and $\delta'$ are defined in Theorem \ref{theorem1_central_nonUniform}.
\vspace{-10pt}
\end{lemma}
The proof is presented in Appendix \ref{appendix::sensitivity_analysis}.

% \begin{figure}[t]
% \centering
% 	\centering
% 	{\includegraphics[width=0.45\linewidth]{composition_theorem_modified.eps}}
%       \caption{Total privacy leakage as a function of $K$, number of users for different values of $T$, the number of training iterations, where $L=1$, $\sigma_{k,t} =N_{0} = 3$, $\delta_{l} = \delta' = 10^{-4}$ and $ p  =  \frac{2}{\sqrt{K}} \sqrt{\frac{1}{2} \log \left(\frac{2}{\delta'}\right)}$. }
%  \label{fig:privacyscaling}
%     \vspace{-10pt}
% \end{figure}

\begin{remark} 
\vspace{-10pt}
From Lemma \ref{lemma1:local_privacy}, we can observe the privacy benefits of wireless gradient aggregation. 
Asymptotically, the local privacy level behaves like $\mathcal{O}(1/\sqrt{1+\kappa_t})$. In contrast, the local privacy achieved by
orthogonal transmission 
% can be shown to be:
% \begin{align}
%     \epsilon^{(k), \text{Orthogonal}}_{\ell,t}  & =   \frac{2 h_{k,t} \alpha_{k,t} L}{\sqrt{h_{k,t}^{2} \alpha_{k,t}^{2} \sigma_{k,t}^{2} + N_{0} }}  \sqrt{2 \log \frac{1.25}{\delta_{\ell}} } \nonumber \\ 
%     & = \frac{2 L}{\sqrt{ \sigma_{k,t}^{2} + N_{0} }}  \sqrt{2 \log \frac{1.25}{\delta_{\ell}} }, 
% \end{align}
% where $\alpha_{k, t} = 1/h_{k,t}$ when channel inversion is used. The local privacy for orthogonal transmission 
scales as a constant, and does not decay with $K$ \cite{seif2020wireless}. 
\vspace{-10pt}
\end{remark}

While Theorem \ref{theorem1_central_nonUniform} shows the per-iteration leakage, we can use advanced composition results for DP using the Gaussian mechanism to obtain the total privacy leakage when the wireless FL algorithm is used for $T$ iterations. When the sampling probability is time-variant, using existing results in \cite{kairouz2015composition}, it can be readily shown that the total leakage over $T$ iterations of the proposed scheme is $(\epsilon_{c}^{(T)},  \delta_{c}^{(T)} )$-DP for $\tilde{\delta} \in (0, 1]$ where $\epsilon_{c}^{(T)}$ and $\delta_{c}^{(T)}$ can be found as follows,
\begin{align}
    \epsilon_{c}^{(T)} & = \sum_{t \in [T]} \frac{(e^{\epsilon_{c,t}} - 1)\epsilon_{c,t}}{(e^{\epsilon_{c,t}} + 1)} + \sqrt{2 \log(1/\tilde{\delta})  \sum_{t \in [T]} \epsilon_{c,t}^{2}},\label{eq:HeteroComp} \\ 
      %\delta_{c}^{(T)} &= 1 - (1-\tilde{\delta}) \prod_{t = 1}^{T} (1 - \delta_{c,t}).\\
    & \overset{(a)} \leq \left( e^{\frac{c}{\sqrt{\min_{t} \mu_{|\mathcal{K}_{t}|} - \beta K}}} -1 \right)^{2} \times \frac{\sum_{t \in [T]} (\max_{k} p_{k,t})^{2}}{2 (1- \delta')^{2}}  \nonumber \\ 
    & \hspace{0.1in} + \sqrt{2 \log(1/\tilde{\delta})} \left( e^{\frac{c}{\sqrt{\min_{t} \mu_{|\mathcal{K}_{t}|} - \beta K}}} -1 \right)\frac{ \sqrt{\sum_{t \in [T]} (\max_{k} p_{k,t})^{2} }}{1- \delta'}, \label{eq:HeteroComp_simplified}
\end{align}
where step $(a)$ follows from the fact that $e^{x} + 1 \geq 2$, where  $x \geq 0$ and $\log(1 + x) \leq x$. Also, 
\begin{align}
    \delta_{c}^{(T)} &= 1 - (1-\tilde{\delta}) \prod_{t = 1}^{T} (1 - \delta_{c,t}) = 1 - (1-\tilde{\delta}) \prod_{t = 1}^{T} \left(1 - \big(\delta' + \frac{\max_k p_{k,t} \delta_{\ell}}{1-\delta'} \big) \right)
\end{align}
By examining the expression in \eqref{eq:HeteroComp_simplified}, we can see that, for a given $T$, $\min_{t} \mu_{|\mathcal{K}_{t}|} - \beta K$ grows as $K$ increases. Therefore, the exponential term approaches $1$ as $K$ increases, and \eqref{eq:HeteroComp_simplified} goes to $0$ as the number of users increases. For the case when the sampling probability is time-invariant, using existing results in \cite{dwork2010boosting}, it can be readily shown that the total leakage over $T$ iterations of the proposed scheme is $(\epsilon_{c}^{(T)}, T \delta_{c} + \tilde{\delta})$-DP for $\tilde{\delta} \in (0, 1]$ where $\epsilon_{c}^{(T)}= \sqrt{2 T \log(1/\tilde{\delta})} \epsilon_{c} + T \epsilon_{c} (e^{\epsilon_{c}} -1)$.
We can expect the same behavior to hold true for the time-invariant case since the result in \cite{kairouz2015composition} is more general than the result in \cite{dwork2010boosting}. We illustrate the total central privacy leakage for the uniform sampling time-invariant case as a function of $K$ in Fig. \ref{fig:privacyscaling} for various values of $T$. As is clearly evident, the leakage provided by wireless FedSGD goes asymptotically to $0$ as $K\rightarrow \infty$. 
%It is worth noting that total privacy leakage over $T$ iterations can be further tightened using existing techniques such as Rényi DP \cite{mironov2017renyi} composition and moment accountant method \cite{abadi2016deep}.

\subsection{Convergence rate of private FL}
In this section, we analyze the performance of  private wireless FedSGD under the assumption that the global loss function $F(\mathbf{w})$ is smooth and strongly convex, and the data across users is i.i.d. Specifically, we consider two scenarios when $(a)$ $\mathcal{K}_t$ is unknown and $(b)$ $\mathcal{K}_t$ is known to the PS. We take both privacy and wireless aggregation into account while deriving the bounds. Interestingly, we show that the unknown $\mathcal{K}_t$ case always outperforms the known $\mathcal{K}_t$ case.
Therefore, it is not necessary for the PS to know $\mathcal{K}_t$. We confirm this observation in the experiment section as well. Due to privacy requirements and noisy nature of wireless channel, the convergence rate is penalized  as shown in the following Theorem. 

\begin{theorem}\label{theorem_1_convergence} 
\vspace{-10pt}
(Unknown $\mathcal{K}_t$ with non-uniform sampling)
Suppose the loss function $F$ is $\lambda$-strongly convex and $\mu$-smooth with respect to $\mathbf{w}^{*}$. Then, for a learning rate $\eta_{t} = 1 / \lambda t$ and a number of iterations $T$, the convergence rate of the private wireless FedSGD algorithm is 
\begin{align}
     \mathds{E} \left[ F(\mathbf{w}_{T}) \right] - F(\mathbf{w}^{*})  \leq  \frac{2 \mu}{ \lambda^{2} T^{2}}  \sum_{t=1}^{T}\left[\frac{L^{2} \big(\mu_{|\mathcal{K}_{t}|}^{2} + \sigma_{|\mathcal{K}_{t}|}^{2}\big) }{\mu_{|\mathcal{K}_{t}|}^{2}} + \frac{d}{ \mu_{|\mathcal{K}_{t}|}^{2} }  \left[  \max_{k} \sigma_{k,t}^{2} \times  \mu_{|\mathcal{K}_{t}|} +N_{0}\right] \right], \label{eq:convergenceBoundNonuniform}
\end{align}
where  $\mu_{|\mathcal{K}_{t}|} = \sum_{k = 1}^{K} p_{k,t}$ and $\sigma_{|\mathcal{K}_{t}|}^{2} = \sum_{k=1}^{K} p_{k,t}(1-p_{k,t})$.
\vspace{-10pt}
\end{theorem}

Theorem \ref{theorem_1_convergence} is proved in Appendix \ref{appendix:proof_theorem_2}. From the above result, we observe that the convergence rate depends on: $(a)$ the total number of users $K$, $(b)$ the number of model parameters $d$, $(c)$ worst amount of perturbation noise across user per iteration, and $(d)$ the sampling probabilities $p_{k,t}$s.  
% We next present a special case of Theorem \ref{theorem_1_convergence}, that is with uniform sampling.
% \begin{corollary} (Unknown $\mathcal{K}_t$ with uniform sampling) Under the same assumptions as Theorem \ref{theorem_1_convergence}, 
% % Suppose the loss function $F$ is $\lambda$-strongly convex and $\mu$-smooth with respect to $\mathbf{w}^{*}$. Then, for a learning rate $\eta_{t} = 1 / \lambda t$ and a number of iterations $T$, 
% the convergence rate for the uniform case is bounded as follows, 
% \begin{align}
%     \hspace{-5pt} \mathds{E} \left[ F(\mathbf{w}_{T}) \right] - F(\mathbf{w}^{*})  \leq  \frac{2 \mu}{ \lambda^{2} T^{2}}  \sum_{t=1}^{T}\left[\frac{L^{2}  ( p_{t} (K-1)+1 )}{K p_{t} } + \frac{d}{ K^{2} p_{t}^{2} } \left[  \max_{k} \sigma_{k,t}^{2} \times  K p_{t} +N_{0}\right] \right]. \label{convergence_equation}
% \end{align}\label{corollary:UnknownUniform}
% \vspace{-20pt}
% \end{corollary}
% It can be readily shown that the convergence rate is monotonically decreasing as a function of $p_{t}$. It can also be seen that bound depends directly on the total number of users $K$.
When the $p_t^*$ from \eqref{eqn:central_optimal_sampling} is used, the convergence rate becomes the following.
\begin{corollary} 
\vspace{-10pt}
(Convergence under optimal $p_{t}^{*}$ from  \eqref{eqn:central_optimal_sampling})
Under the same assumptions as Theorem \ref{theorem_1_convergence},
% Suppose the loss function $F$ is $\lambda$-strongly convex and $\mu$-smooth with respect to $\mathbf{w}^{*}$. Then, for a learning rate $\eta_{t} = 1 / \lambda t$ and a number of iterations $T$, 
the convergence rate for the case when the optimal sampling probability $p_{t}^*$ from \eqref{eqn:central_optimal_sampling}  is
\begin{align}
     \hspace{-7pt} \mathds{E} \left[ F(\mathbf{w}_{T}) \right] - F(\mathbf{w}^{*})  %& \leq  \frac{2 \mu}{ \lambda^{2} T^{2}}  \sum_{t=1}^{T}\left[\frac{L^{2} (  \alpha (\sqrt{K}-1/\sqrt{K})+1 )}{\alpha \sqrt{K} } + \frac{d}{ \alpha^{2} K } \left[  \max_{k} \sigma_{k,t}^{2} \times  \alpha \sqrt{K} +N_{0}\right] \right] \nonumber \\ 
     &  \leq  \frac{2 \mu}{ \lambda^{2} T} \left[\frac{L^{2} (  \alpha (\sqrt{K}-1/\sqrt{K})+1 )}{\alpha \sqrt{K} } + \frac{d}{ \alpha^{2} K } \left[ \alpha \sqrt{K} \max_{k,t} \sigma_{k,t}^{2}   +N_{0}\right] \right],
     \label{convergence_equation_p_optimal}
\end{align}
where $\alpha =  2 \sqrt{\frac{1}{2} \log \frac{2}{\delta' }}$.
\vspace{-10pt}
\end{corollary}
It can be seen that the constant in front of both bounds scale as $\mathcal{O}(1/T)$. However, the second parts of the expressions depends on the sampling probabilities. We can see from \eqref{convergence_equation_p_optimal} that the first term in the bracket is constant and that the second term scales as $\mathcal{O}(1/\sqrt{K})$. Since $p_t^*$ is obtained when privacy is prioritized, \eqref{convergence_equation_p_optimal} is potentially the worst bound of the two. One can potentially select sampling probabilities for \eqref{eq:convergenceBoundNonuniform} to obtain even better scaling than $\mathcal{O}(1/\sqrt{K})$. We next present the convergence results for the case when $\mathcal{K}_t$ is known at the PS.
 
\begin{theorem}\label{theorem_1_convergence_known_kt} 
\vspace{-10pt}
(Known $\mathcal{K}_t$ with non-uniform sampling)
Suppose the loss function $F$ is $\lambda$-strongly convex and $\mu$-smooth with respect to $\mathbf{w}^{*}$. Then, for a learning rate $\eta_{t} = 1 / \lambda t$ and a number of iterations $T$, the convergence rate of the private wireless FedSGD algorithm is given as
\begin{align}
     \mathds{E} \left[ F(\mathbf{w}_{T}) \right] - F(\mathbf{w}^{*})  \leq  \frac{2 \mu}{ \lambda^{2} T^{2}}  \sum_{t=1}^{T}\left[\frac{L^{2} }{\zeta_t} + \frac{d}{\zeta_{t}^{2} }\left[ \max_{k} \sigma_{k,t}^{2} \times  \mathds{E}\left[ \frac{1}{|\mathcal{K}_{t}|}\right] + \mathds{E} \left[ \frac{1}{{|\mathcal{K}_{t}|}^{2}} \right] N_{0}\right] \right], \label{convergence_equation2}
\end{align}
where $\zeta_{t} = 1 - \prod_{k=1}^{K} (1-p_{k,t})$.
%where    $c^{2} =  \frac{ \min_{j}  |h_{j}|^{2} P_{j}}{L^{2}}$.
\vspace{-10pt}
\end{theorem}

% \begin{corollary}\label{cor_convergence_known_kt} (Known $\mathcal{K}_t$ with uniform sampling)
% Under same assumptions as Theorem \ref{theorem_1_convergence_known_kt}, the convergence rate for the uniform case is bounded by \eqref{convergence_equation2} with $\zeta_{t} = 1 -  (1-p_{t})^{K}$.
% \end{corollary}
% The bounds in 
Theorem \ref{theorem_1_convergence_known_kt} depends on $\mathds{E}\left[ \frac{1}{|\mathcal{K}_{t}|}\right]$ and $\mathds{E} \left[ \frac{1}{{|\mathcal{K}_{t}|}^{2}} \right]$. Note that $\mathcal{K}_t$ is a binomial random variable. It is difficult to obtain closed form expressions for $\mathds{E}\left[ \frac{1}{|\mathcal{K}_{t}|}\right]$ and $\mathds{E} \left[ \frac{1}{{|\mathcal{K}_{t}|}^{2}} \right]$. However, it is possible to approximate them using Taylor series approximation, specifically,
we approximate $\mathds{E} \left[ \frac{1}{{|\mathcal{K}_{t}|}} \right]$ using Taylor's series around $\mathds{E}\left[|\mathcal{K}_{t}| \right]$ for upto second degree as follows:
 \begin{align}
     \mathds{E} \left[ \frac{1}{{|\mathcal{K}_{t}|}} \right] &\approx \mathds{E} \left[\frac{1}{\mathds{E}\left[|\mathcal{K}_{t}| \right]} - \frac{1}{\mathds{E}^{2}\left[|\mathcal{K}_{t}| \right]} \big(|\mathcal{K}_{t}| - \mathds{E}\left[|\mathcal{K}_{t}| \right] \big) + \frac{1}{\mathds{E}^{3}\left[|\mathcal{K}_{t}| \right]} (|\mathcal{K}_{t}| - \mathds{E}\left[|\mathcal{K}_{t}| \right])^{2} \right] \nonumber \\ 
     & = \frac{1}{\mu_{|\mathcal{K}_{t}|} } + \frac{\sigma_{|\mathcal{K}_{t}|}^{2}}{\mu_{|\mathcal{K}_{t}|}^{3}}.\label{eq:approximation}
 \end{align}
 Similarly for $\mathds{E} \left[ \frac{1}{{|\mathcal{K}_{t}|}^{2}} \right]$, we approximate it around $\mathds{E}\left[|\mathcal{K}_{t}| \right]$ as follows: 
 \begin{align}
       \mathds{E} \left[ \frac{1}{{|\mathcal{K}_{t}|}^{2}} \right] \approx  \frac{1}{\mu_{|\mathcal{K}_{t}|}^{2} } + \frac{3\sigma_{|\mathcal{K}_{t}|}^{2}}{\mu_{|\mathcal{K}_{t}|}^{4}}.\label{eq:approximation2}
 \end{align}
 
By plugging \eqref{eq:approximation} and \eqref{eq:approximation2} back to Theorem \ref{theorem_1_convergence_known_kt} for the uniform sampling case, and setting $p_{k,t}=p,\forall k,t$, $\mu_{|\mathcal{K}_{t}|}^{2}=Kp$ and $\sigma_{|\mathcal{K}_{t}|}^{2}=Kp(1-p),\forall t$, we obtain,
\begin{align}
     \mathds{E} \left[ F(\mathbf{w}_{T}) \right] - F(\mathbf{w}^{*})  \leq  \frac{2 \mu}{ \lambda^{2} T} \left[\frac{L^{2} }{\zeta} + \frac{d}{Kp\zeta^{2} } \left[ \max_{k,t} \sigma_{k,t}^{2} \times  (1 + (1-p)^2) + (1 + 3(1-p)^2) \frac{N_{0}}{Kp}\right] \right], \nonumber
\end{align}
where $\zeta=1-(1-p)^K$.

\begin{figure}[t]
\centering
    \begin{minipage}{.45\textwidth}
	\centering
	 \subcaptionbox{Theoretical bounds on optimality gap.}
	{\includegraphics[width=.95\linewidth]{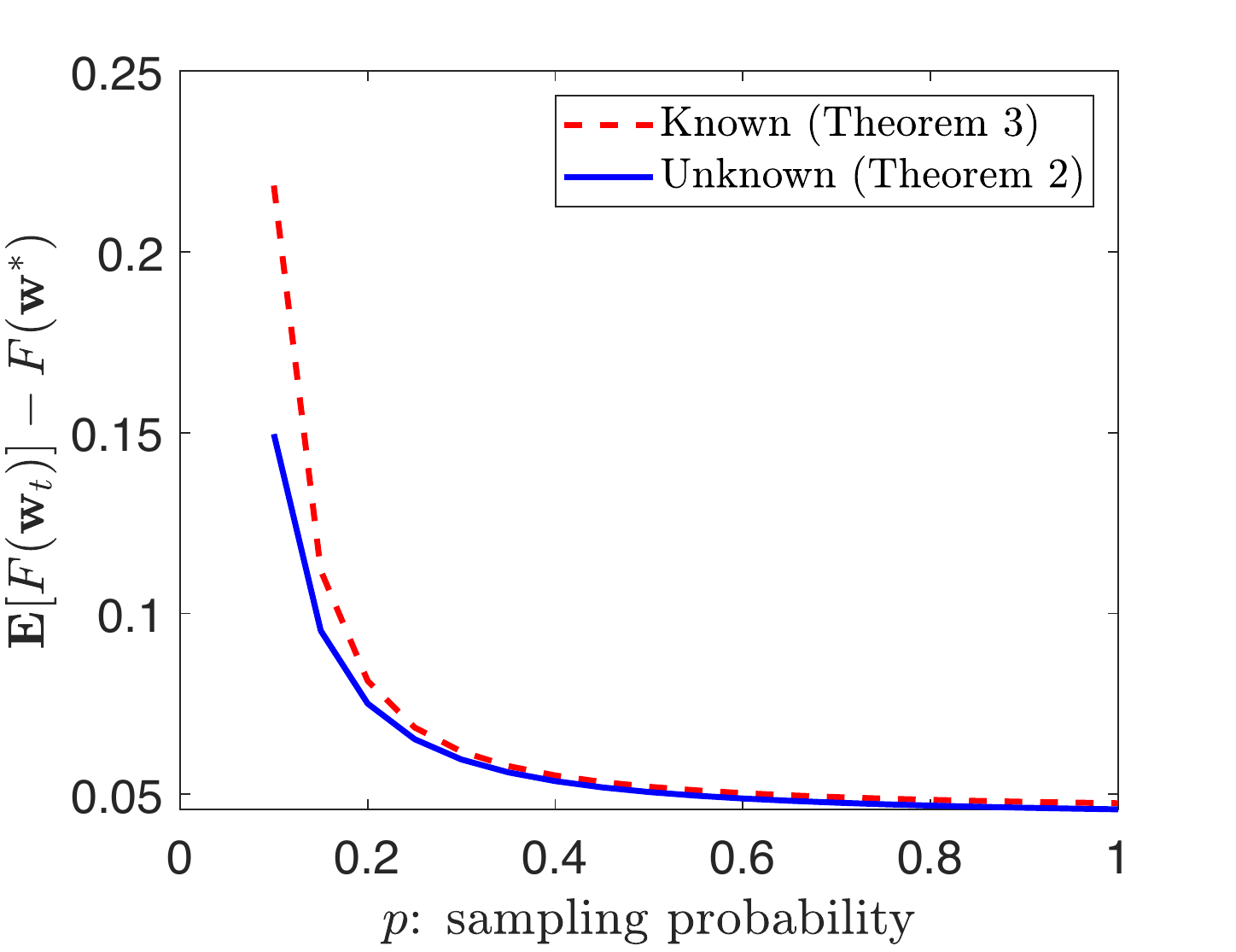}}
	\end{minipage}
	\begin{minipage}{.45\textwidth}
	\centering
	 \subcaptionbox{Empirical.}
	{\includegraphics[width=\linewidth]{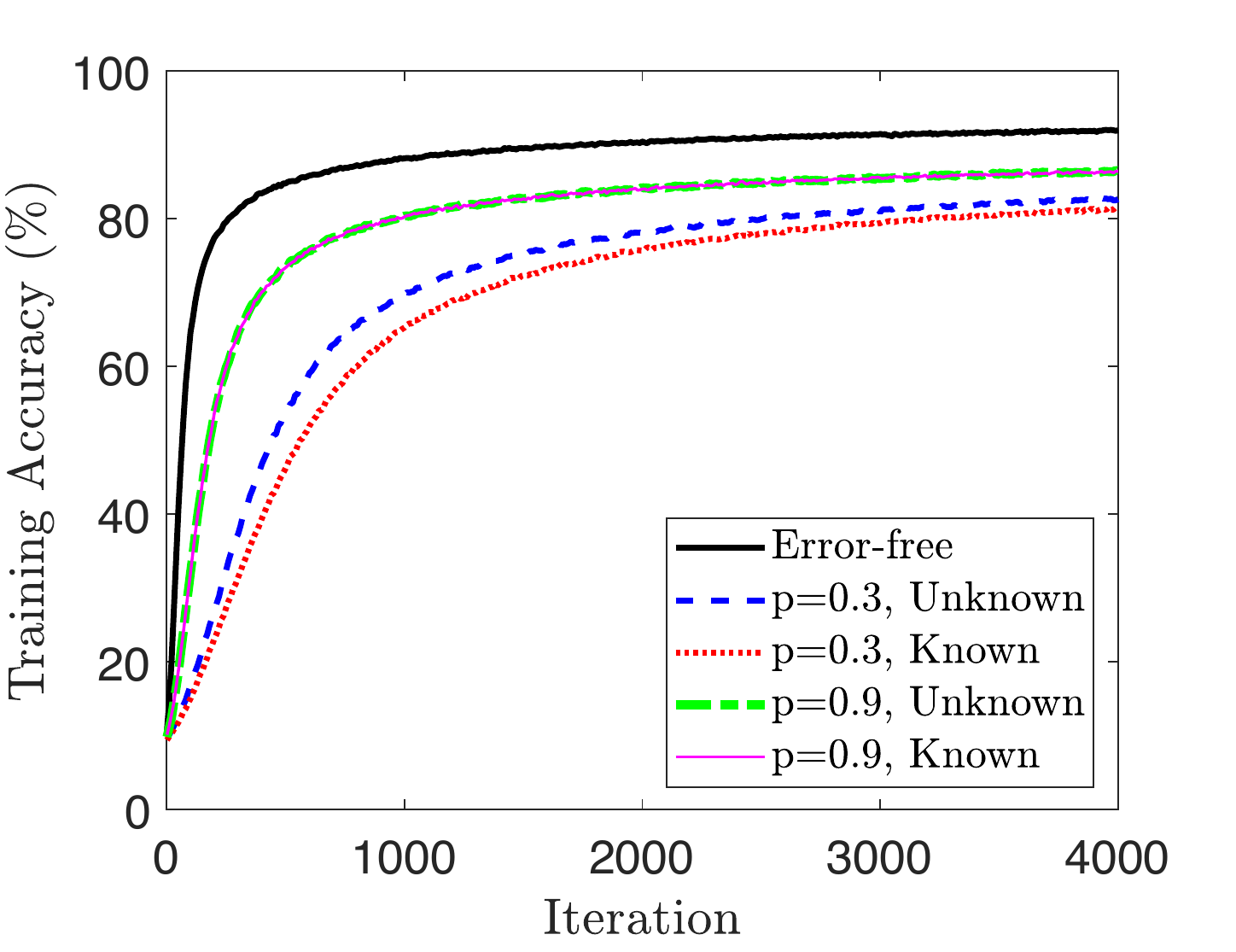}}
	\end{minipage}
	\vspace{-15pt}
    \caption{Comparisons of convergence bounds and training accuracy with uniform sampling for both cases: $(1)$ unknown $\mathcal{K}_{t}$ or $(2)$ known $\mathcal{K}_{t}$, where $K=20$, $L=2$, $T = 4000$, $\lambda = 0.2$, $\mu = 0.9$, $d =30$, $N_{0} = 1$, $\sigma_{k, t}^2=0.1$ and $\delta_{\ell} = \delta' = 10^{-5}$. Each user has transmit $\text{SNR}_k=10$ dB and $(b)$ is trained on MNIST dataset.}
    \vspace{-30pt}
    \label{fig:theoretical_bounds_convergence}
\end{figure}

We note that this bound behaves similarly to the bound in Theorem \ref{theorem_1_convergence} with $p_{k,t}=p,\forall k,t$ when either $T$ or $K$ is large.
Therefore, the proposed scheme performs similarly when $\mathcal{K}_t$ is known or unknown. This can be seen in Fig. \ref{fig:theoretical_bounds_convergence} where the curves are obtained for $K=200$ users, and $T=4000$ iterations. We also show this empirically in Fig. \ref{fig:theoretical_bounds_convergence} using MNIST dataset. It can be seen that for the same sampling probability $p$, schemes with unknown $\mathcal{K}_t$ are always better than schemes with known $\mathcal{K}_t$. The difference between two approaches is only at the scaling of the aggregated gradient. This observation indicates that as long as the direction of the aggregated gradient is preserved and the scaling is not drastically different, the performance of the SGD algorithm will not deviate much \cite{ajalloeian2020analysis}. This is due to the fact that the magnitude of the gradient at a particular iteration is always corrected in the following iterations as long as the direction is correct. Therefore, it might not be necessary to ask users to coordinate among themselves to preserve privacy as claimed in \cite{hasircioglu2020private}.

% \begin{remark}\label{remark:centralBehavior}
% We notice that the behavior of the central privacy obtained in \eqref{eq:uniform_central_privacy} depends on the number of users, perturbation noises and the channel noise, which implicitly depends on the constant $c$ (defined in \eqref{eq:uniform_central_privacy}). In Fig. \ref{fig:central_behavior}, we plot the central privacy versus the sampling probability for $K = 20$ users and different values of $c$. We note that when $K$ and/or perturbation noises are large, the central privacy always has similar behavior as the curve obtained from $c=2$ in Fig. \ref{fig:central_behavior}. This can also be observed from \eqref{eq:uniform_central_privacy}. The central privacy depends on the product of the sampling probability $p_t$ and an exponential term that is a function of $p_t$. For small fixed $K$ and $\sigma_{k,t}^2$ (hence, large $c$), the exponential term becomes the dominating term. The exponential term decreases at a much faster rate compared to the rate $p_t$ can increase. Therefore, the central privacy leakage becomes convex for larger $c$.
% \end{remark}

%Experiments:

\section{Experiments} \label{experiments}
In this section, we conduct experiments to assess the performance of the  wireless FedSGD with user sampling on MNIST dataset for image classification. We model the instances of fading channels $h_{k,t}$'s via an autoregressive (AR) Rician model \cite{tse2005fundamentals}, where the Rician parameter $\Gamma=5$ and the temporal correlation coefficient $\rho=0.1$. 
The channel noise variance (receiver noise) is set as $N_{0}=1$.  The user's transmit signal-to-noise ratio is defined as $\text{SNR}_{k} = \frac{P_{k}}{d N_{0}} $.
We use $\sigma_{k,t}^2=0.1$ as the perturbation noise. 
Prior to sending the local gradient to the PS, each user clips the local gradient using the Lipschitz constant chosen empirically with test runs. We use $\delta_{\ell}=10^{-5}$ and $\delta'=2e^{-2\mu_{|\mathcal{K}_{t}|}^2/K}+10^{-5}$ to satisfy the constraint on $\delta'$ and to avoid it from going to $0$. We consider two different sampling schemes described as follows,

\begin{figure}[t]
\centering
	\begin{minipage}{.45\textwidth}
	\centering
	 \subcaptionbox{$L=1$, $T=400$.\label{fig:MNISTiidUniformSingleL1}}
	{\includegraphics[width=\linewidth]{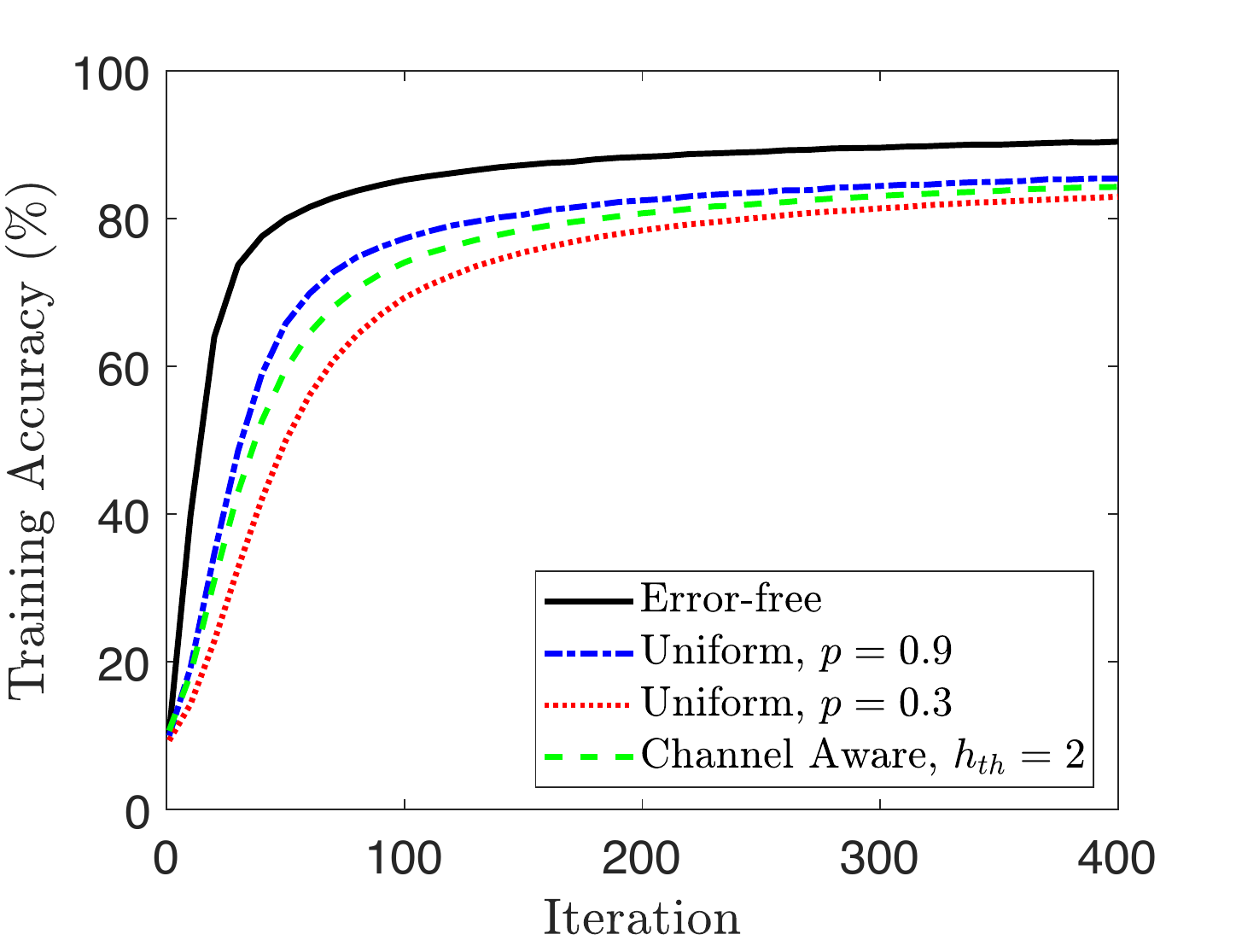}}
	%\vspace{-1pt}
% 	\caption{The impact of the sampling probability on the training accuracy for the case where sampling probabilities are the same across users and iterations.
	
	%\vspace{-15pt}
	\end{minipage}
	\begin{minipage}{.45\textwidth}
	\centering
	 \subcaptionbox{$L=0.1$, $T=2500.$\label{fig:MNISTiidUniformSingleL01}}
	{\includegraphics[width=\linewidth]{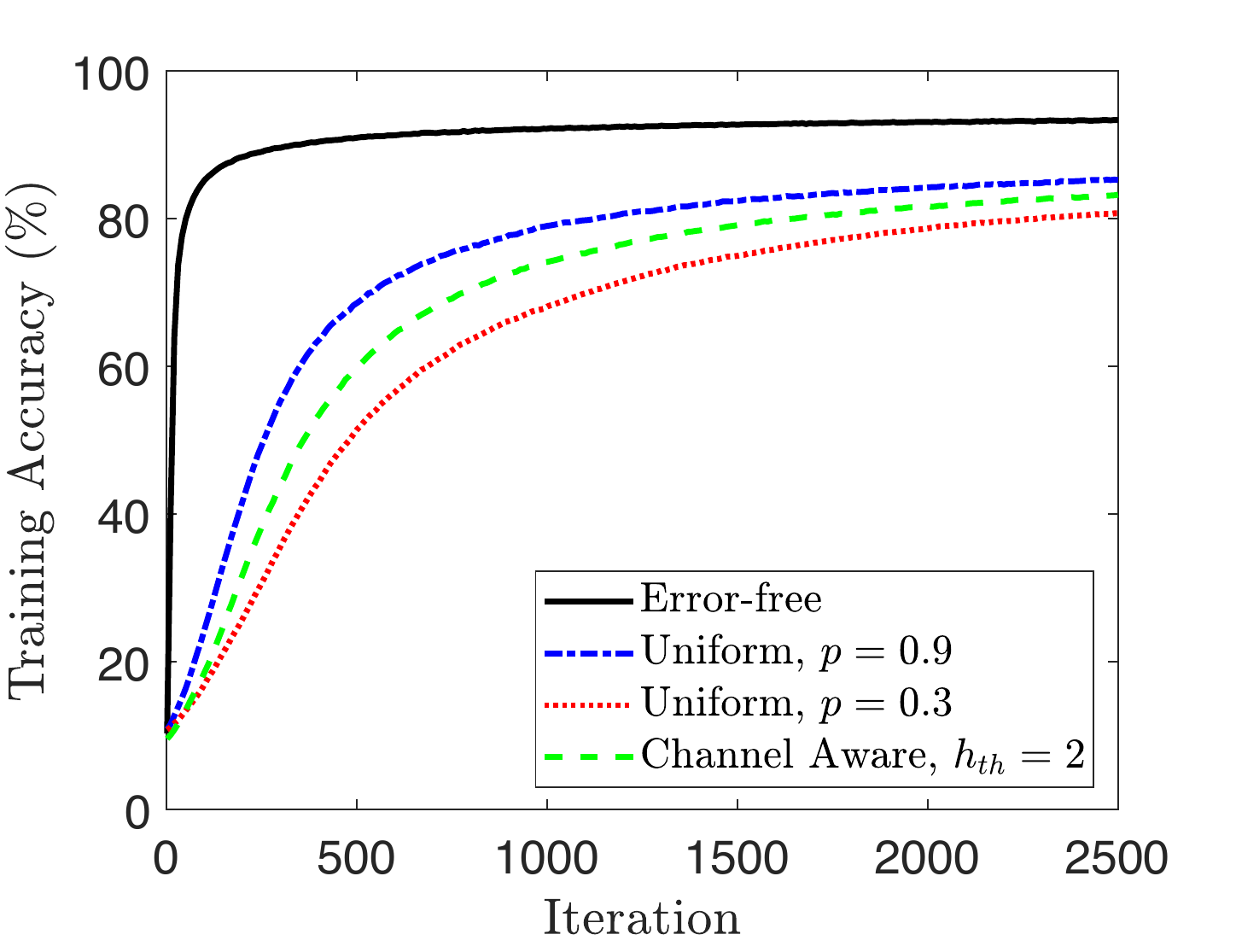}}
	%\vspace{-1pt}
	%\caption{
	%\vspace{-15pt}
	\end{minipage}
	\vspace{-15pt}
	\caption{The impact of the sampling probability on the training accuracy for single-layer neural network trained on MNIST dataset with $\sigma_{k,t}^2=0.1$.
	\label{fig:MNISTUniformSingle}
	}
	\vspace{-10pt}
\end{figure}
\begin{table}[t]
 \centering
  \begin{tabular}{|c|c|c|c|}
         \hline
     & Channel Aware & \multicolumn{2}{c|}{Uniform}\\
     \cline{2-4}
     & $h_{\text{th}}=2$ & $p=0.3$ & $p=0.9$\\
     \hline
$\epsilon_{\ell,\max}$  & $3.675$ & $5.124$ & $2.46$\\
$\epsilon_{c,\max}$  & $4.535$ & $5.61$ & $3.132$\\
Avg. $|\mathcal{K}|$ & $96$ & $60$ & $180$\\
Testing Acc.  & $85.27\%$ & $83.98\%$ & $86.42\%$\\
         \hline
         \multicolumn{4}{c}{\textbf{(a)} $L=1,T=400.$}
 \end{tabular}
 \quad
 \begin{tabular}{|c|c|c|c|}
         \hline
     & Channel Aware & \multicolumn{2}{c|}{Uniform}\\
     \cline{2-4}
     & $h_{\text{th}}=2$ & $p=0.3$ & $p=0.9$\\
    \hline
$\epsilon_{\ell,\max}$  & $0.3677$ & $0.5124$ & $0.2460$\\
$\epsilon_{c,\max}$  & $0.3642$ & $0.2258$ & $0.2317$\\
Avg. $|\mathcal{K}|$ & $96$ & $60$ & $180$\\
Testing Acc.  & $84.33\%$ & $81.76\%$ & $86.25\%$\\
    \hline
    \multicolumn{4}{c}{\textbf{(b)} $L=0.1,T=2500.$}
 \end{tabular}
 \vspace{-10pt}
 \caption{Comparison of privacy leakage per iteration for single-layer neural network with $\sigma_{k,t}^2 = 0.1$. $\epsilon_{\ell, \max}$ and $\epsilon_{c, \max}$ denote the maximum local and central leakages across iterations, respectively.
 \label{tbl:estimate}}
 \vspace{-30pt}
\end{table}

\begin{figure}[t]
\centering
	\begin{minipage}{.45\textwidth}
	\centering
	 \subcaptionbox{$L=1$.\label{fig:MNISTiidUniformTwoL1}}
	{\includegraphics[width=\linewidth]{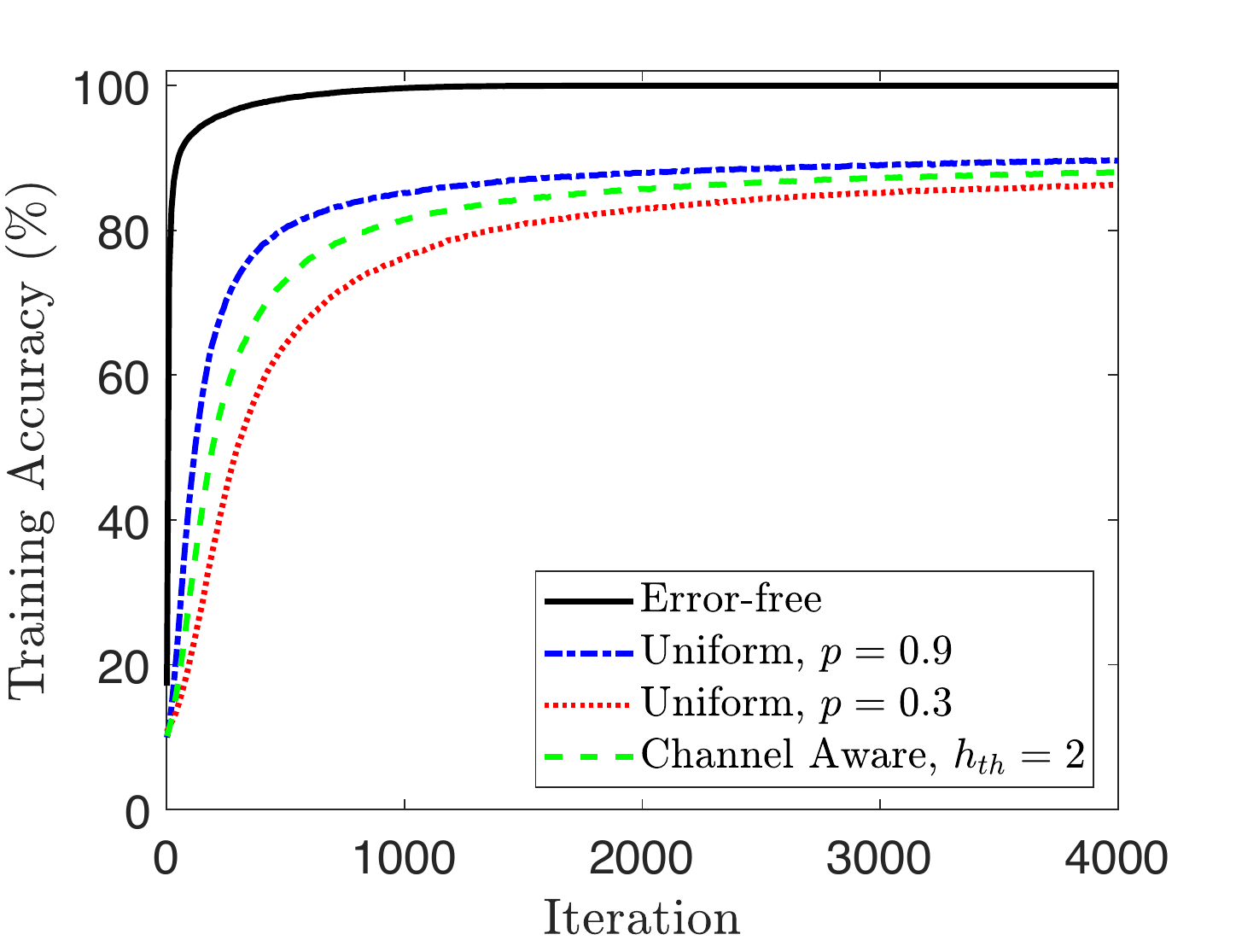}}
	%\vspace{-1pt}
% 	\caption{The impact of the sampling probability on the training accuracy for the case where sampling probabilities are the same across users and iterations.
	
	%\vspace{-15pt}
	\end{minipage}
	\begin{minipage}{.45\textwidth}
	\centering
	 \subcaptionbox{$L=0.2$.\label{fig:MNISTiidUniformTwoL02}}
	{\includegraphics[width=\linewidth]{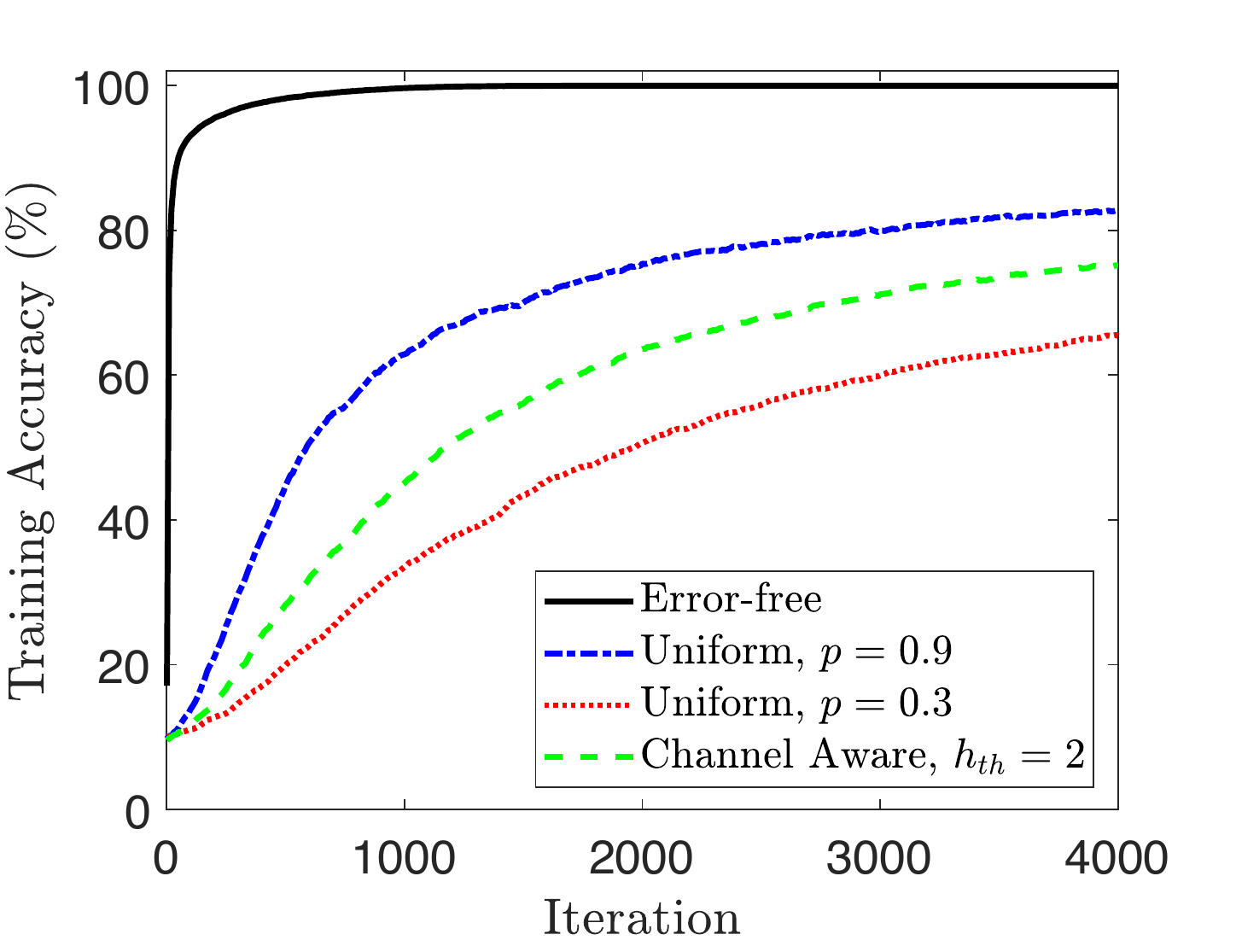}}
	%\vspace{-1pt}
	%\caption{
	%\vspace{-15pt}
	\end{minipage}
	\vspace{-10pt}
	\caption{The impact of the sampling probability on the training accuracy for two-layer neural network trained on MNIST dataset with $\sigma_{k,t}^2=0.8$.
	\label{fig:MNISTUniformTwo}
	}
	\vspace{-10pt}
\end{figure}

\begin{table}[t]
 \centering
  \begin{tabular}{|c|c|c|c|}
         \hline
     & Channel Aware & \multicolumn{2}{c|}{Uniform}\\
     \cline{2-4}
     & $h_{\text{th}}=2$ & $p=0.3$ & $p=0.9$\\
     \hline
$\epsilon_{\ell,\max}$  & $1.390$ & $2.084$ & $0.8953$\\
$\epsilon_{c,\max}$  & $1.991$ & $1.653$ & $1.487$\\
Avg. $|\mathcal{K}|$ & $96$ & $60$ & $180$\\
Testing Acc.  & $88.72\%$ & $87.10\%$ & $90.28\%$\\
         \hline
         \multicolumn{4}{c}{\textbf{(a)} $L=1.$}
 \end{tabular}
 \quad
 \begin{tabular}{|c|c|c|c|}
         \hline
     & Channel Aware & \multicolumn{2}{c|}{Uniform}\\
     \cline{2-4}
     & $h_{\text{th}}=2$ & $p=0.3$ & $p=0.9$\\
    \hline
$\epsilon_{\ell,\max}$  & $0.2795$ & $0.4169$ & $0.1791$\\
$\epsilon_{c,\max}$  & $0.2620$ & $0.1505$ & $0.1633$\\
Avg. $|\mathcal{K}|$ & $96$ & $60$ & $180$\\
Testing Acc.  & $75.89\%$ & $66.33\%$ & $83.68\%$\\
    \hline
    \multicolumn{4}{c}{\textbf{(b)} $L=0.2.$}
 \end{tabular}
 \vspace{-10pt}
 \caption{Comparison of privacy leakage per iteration for two-layer neural network with $\sigma_{k,t}^2 = 0.8$.
 \label{tbl:estimate2}}
 \vspace{-30pt}
\end{table}

\begin{figure}[t]
\centering
	\begin{minipage}[t]{.45\textwidth}
	\centering
% 	 \subcaptionbox{Single}
	{\includegraphics[width=0.98\linewidth]{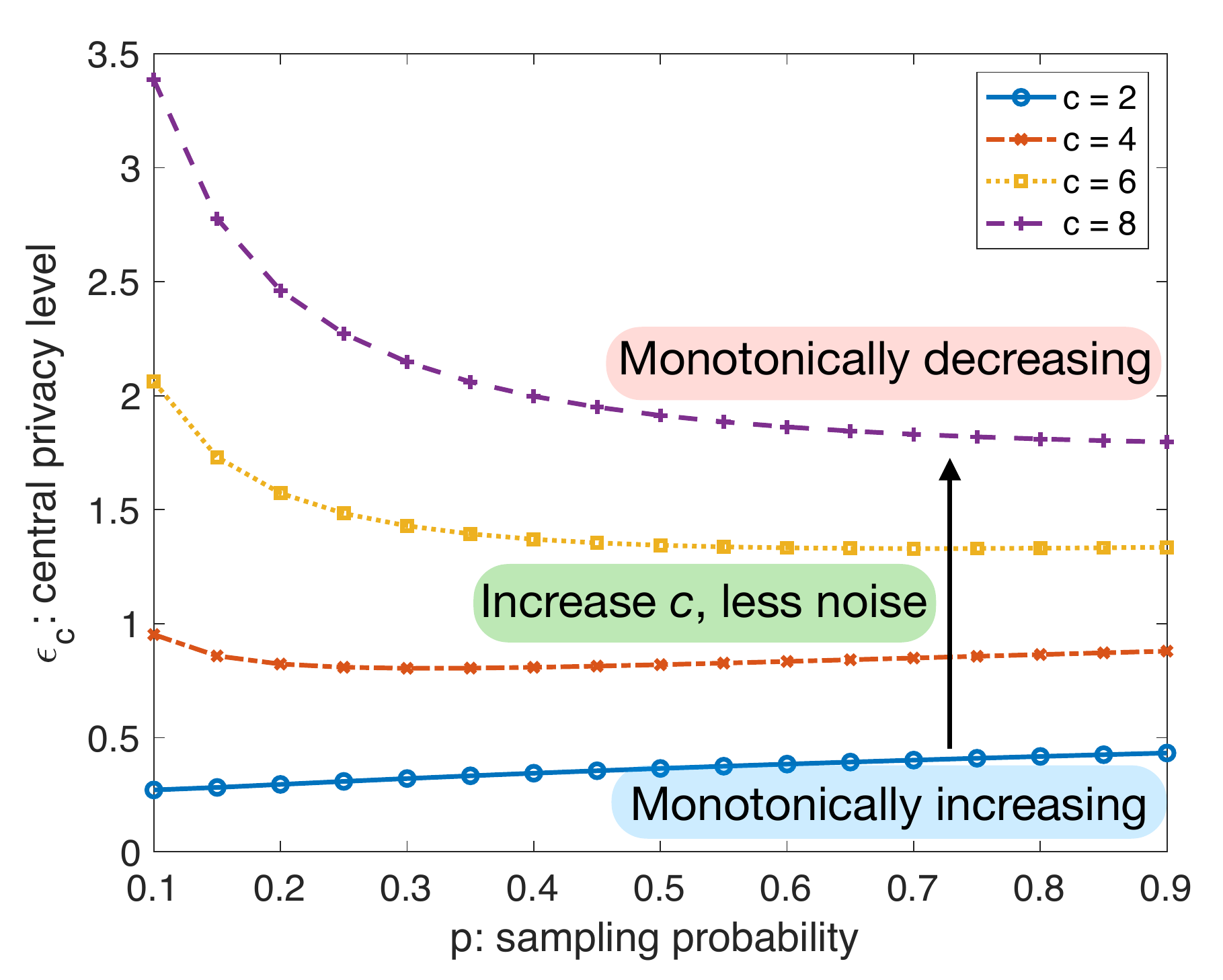}}
% 	\vspace{-10pt}
	\caption{Central DPs as a function of $p$ for different values of $c$ with $K=20$. 
	\label{fig:epsilonCBehavior}}
	%\vspace{-15pt}
	\end{minipage}
	\quad
	\begin{minipage}[t]{.45\textwidth}
	\centering
% 	 \subcaptionbox{Two}
	{\includegraphics[width=\linewidth]{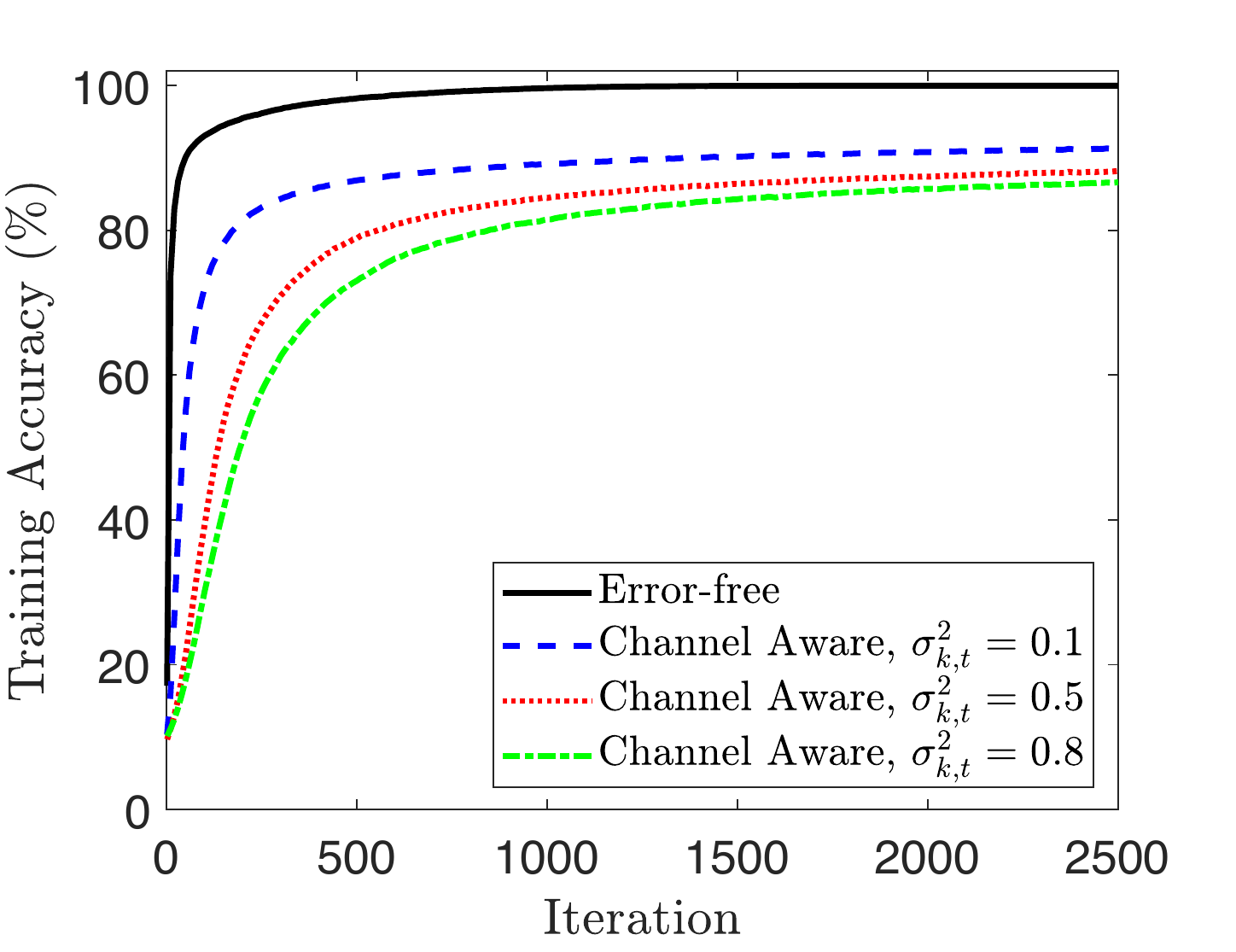}}
% 	\vspace{-10pt}
	\caption{The impact of the perturbation noise on training accuracy for two-layer NN. \label{fig:HiddenLayerNoiseImpact}}
	%\vspace{-15pt}
	\end{minipage}
% 	\caption{The impact of the sampling probability on the training accuracy for single-layer neural network trained on MNIST dataset for the case where sampling probabilities are the same across users and iterations.
% 	\label{fig:MNISTUniform}}
\vspace{-30pt}
\end{figure}

% \begin{table}[t]
%  \centering
%   \begin{tabular}{|c|c|c|c|}
%          \hline
%      & Channel Aware & \multicolumn{2}{c|}{Uniform}\\
%      \cline{2-4}
%      & $h_{\text{th}}=2$ & $p=0.3$ & $p=0.9$\\
%      \hline
% $\epsilon_{\ell,\max}$  & $0.3677$ & $0.5124$ & $0.2460$\\
% $\epsilon_{c,\max}$  & $0.3642$ & $0.2258$ & $0.2317$\\
% Avg. $|\mathcal{K}|$ & $96$ & $60$ & $180$\\
% Testing Acc.  & $84.33\%$ & $81.76\%$ & $86.25\%$\\

%          \hline
         
%  \end{tabular}
% %  \vspace{5pt}
%  \caption{Comparison of privacy leakage per iteration with $\sigma_{k,t}^2 = 0.1$, $L=0.1$ and $T=2500$ iterations.
%  \label{tbl:estimate2}}
%  \vspace{-15pt}
% \end{table}

% \begin{figure}[t]
% \centering
%     \begin{minipage}{.45\textwidth}
% 	\centering
% 	 \subcaptionbox{i.i.d. Data}
% 	{\includegraphics[width=\linewidth]{3SNR_H_iid_nonuniform2.pdf}}
% 	%\vspace{-1pt}
% 	%\caption{
% 	%\label{fig:MNISTiidNonuniformHidden}}
% 	%\vspace{-15pt}
% 	\end{minipage}
% 	\begin{minipage}{.45\textwidth}
% 	\centering
% 	 \subcaptionbox{Non-i.i.d. Data}
% 	{\includegraphics[width=\linewidth]{3SNR_H_noniid_nonuniform2.pdf}}
% 	%\vspace{-1pt}
% 	%\caption{
% 	%\label{fig:MNISTNoniidNonuniformHidden}}
% 	%\vspace{-15pt}
% 	\end{minipage}
% 	\caption{Training accuracy for non-uniform schemes for the two-layer neural network architecture (with hidden layer) with $\sigma_{k,t}^2=0.1,0.8$.
% 	\label{fig:MNISTNonuniformHidden}}
% \end{figure}

\textbf{Uniform Sampling:} Let $p_{k, t}=p,~\forall k, t$ for any $p$. 

\textbf{Channel Aware Sampling:} Each user computes $p_{k, t}= h_{k,t}/h_{\text{th}}$, where the threshold $h_{\text{th}}$ is a hyperparameter which is optimized via cross-validation. 

We train two models: $(a)$ a single-layer neural network (NN) (with no hidden layer) and $(b)$ a two-layer NN (with one hidden layer), using MNIST dataset, which consists of $60,000$ training and $10,000$ testing samples. The loss function we used is cross-entropy, and ADAM optimizer for training with a learning rate of $\eta=0.001$. The training samples are evenly and randomly distributed across $K=200$ users. Users are split into three groups where the first group  consists of $68$ users with $\text{SNR}_k=2$ dB; the second and third group consist of $66$ users in each group with $\text{SNR}_k=10$ and $30$ dB, respectively. We use $h_{\text{th}}=2$ as the threshold for the channel aware sampling scheme. Empirically, the scaling factor is computed as follows,
\begin{align}
\alpha_{k,t} = \min \left[ \frac{1}{h_{k,t}}, \frac{\sqrt{P_{k}}}{\sqrt{\Vert \mathbf{g}_{k}(\mathbf{w}_{t}) \Vert^{2} + d \sigma_{k,t}^{2}}} \right].
\end{align}

In Fig. \ref{fig:MNISTUniformSingle} and \ref{fig:MNISTUniformTwo}, we show the impact of sampling probability on the training accuracy. First, we observe that a higher $p$ leads to a higher accuracy for the model. Next, in Table \ref{tbl:estimate}$(a)$, we observe that, for the uniform case with $L=1$, the central DP leakage decreases as $p$ increases, which contradicts with the intuition that higher $p$ leads to higher leakage. However, let $p_{k,t}=p,\forall k,t$ in \eqref{eq:centralNonUni}, i.e.,
\begin{align}
    \epsilon_{c, t} &\leq \log \left[ 1 + \frac{p}{1-\delta'} \left( e^{\frac{c}{\sqrt{K(p - \beta)}}} -1 \right)  \right],
\end{align}
we can see that the behavior of $\epsilon_{c,t}$ depends on two terms: $p/(1-\delta')$ and $\exp(c/\sqrt{K(p - \beta)})$. As $p$ increases, the first term increases and the second term decreases. For a certain range of $c$, the second term dominates, therefore, $\epsilon_{c,t}$, as a whole, decreases. This is due to the fact that, since perturbation noises get aggregated over the wireless channel, the privacy is enhanced. Hence, users are encouraged to participate more when $c$ belongs to this range. In general, $c$ depends on $\sigma_{k,t}, L, \delta_{\ell}$, and $c$ for Fig. \ref{fig:MNISTiidUniformSingleL1} and Table \ref{tbl:estimate}$(a)$ falls in the range that allows the second term to dominate as $p$ increases. We also demonstrate the case when the first term dominates, i.e., $L=0.1$ for this set of parameters. We can see that the central DP leakage increases as $p$ increases from Table \ref{tbl:estimate}$(b)$. When $c$ is in this range, the amplification of privacy is not enough to outweigh the disadvantage of participating more. Thus, the intuition that higher $p$ leads to higher leakage holds. This can also be seen in Fig. \ref{fig:epsilonCBehavior} that the first term dominates when $c=2$ and the second term dominates when $c=4,6,8$. Similar trends can be found in Table \ref{tbl:estimate2}.

From Table \ref{tbl:estimate}, we can also see that channel aware sampling achieves $85.27\%$ and $84.33\%$ testing accuracy, which is lower than those of uniform sampling with $p=0.9$. This is due to the choice of $h_{\text{th}}$. By reducing $h_{\text{th}}$, we can improve the accuracy of the channel aware sampling.
Another interesting observation is that, while channel aware sampling suffers slightly from higher central DP leakages, it does achieve relatively high testing accuracy and low LDP leakage with significant less average number of participants compare to uniform sampling with $p=0.9$.

\section{Conclusion \& Future Directions }
\label{conclusion}
In this work, we showed the privacy benefits of user sampling and wireless aggregation for federated learning. More specifically, we showed that for certain settings (when $c$ is relatively small), the benefit of user sampling outweighs the advantage of wireless aggregation, therefore, creating tension between central DP, local DP and convergence rate. To minimize central DP, user sampling is essential, and we can tradeoff local DP and convergence rate for central DP by sampling less. However, for other settings (when $c$ is relatively large), the privacy amplification from wireless aggregation outweighs the disadvantage of additional leakage from sampling more, making the tension between central DP, local DP and convergence rate disappear. Hence, user sampling is, in fact, discouraged to minimize central DP. The resulting leakage for central DP was shown to scale as $O(1/K^{3/4})$, improving upon prior results on this topic. We also showed that knowing only the statistics of the number of participants at each iteration is at least as good as knowing the exact number of participants and hence eliminating the need for coordination between the PS and users. As a future work, one immediate direction would be to study other variations of FL such as FedAvg, where each user performs local model updates through multiple SGD computations, followed by model exchange with the PS. Another interesting direction would be to consider scenarios where the sampling probabilities can depend on the local gradients/losses. 
% These scenarios may require new techniques for privacy analysis than the ones used in this paper, where sampling probabilities are independent of the local data (gradients/local loss function).

\bibliographystyle{IEEEtran}
\bibliography{myreferences}

% Generated by IEEEtran.bst, version: 1.14 (2015/08/26)
\begin{thebibliography}{10}
\providecommand{\url}[1]{#1}
\csname url@samestyle\endcsname
\providecommand{\newblock}{\relax}
\providecommand{\bibinfo}[2]{#2}
\providecommand{\BIBentrySTDinterwordspacing}{\spaceskip=0pt\relax}
\providecommand{\BIBentryALTinterwordstretchfactor}{4}
\providecommand{\BIBentryALTinterwordspacing}{\spaceskip=\fontdimen2\font plus
\BIBentryALTinterwordstretchfactor\fontdimen3\font minus
  \fontdimen4\font\relax}
\providecommand{\BIBforeignlanguage}[2]{{%
\expandafter\ifx\csname l@#1\endcsname\relax
\typeout{** WARNING: IEEEtran.bst: No hyphenation pattern has been}%
\typeout{** loaded for the language `#1'. Using the pattern for}%
\typeout{** the default language instead.}%
\else
\language=\csname l@#1\endcsname
\fi
#2}}
\providecommand{\BIBdecl}{\relax}
\BIBdecl

\bibitem{mcmahan2017communication}
B.~McMahan, E.~Moore, D.~Ramage, S.~Hampson, and B.~A. y~Arcas,
  ``Communication-efficient learning of deep networks from decentralized
  data,'' in \emph{Artificial Intelligence and Statistics}, 2017, pp.
  1273--1282.

\bibitem{shokri2017membership}
R.~Shokri, M.~Stronati, C.~Song, and V.~Shmatikov, ``Membership inference
  attacks against machine learning models,'' in \emph{2017 IEEE Symposium on
  Security and Privacy (S $\&$ P)}, May 2017, pp. 3--18.

\bibitem{hayes2019logan}
J.~Hayes, L.~Melis, G.~Danezis, and E.~De~Cristofaro, ``{LOGAN}: Membership
  inference attacks against generative models,'' \emph{Proceedings on Privacy
  Enhancing Technologies}, vol. 2019, no.~1, pp. 133--152, 2019.

\bibitem{melis2019exploiting}
L.~Melis, C.~Song, E.~De~Cristofaro, and V.~Shmatikov, ``Exploiting unintended
  feature leakage in collaborative learning,'' in \emph{2019 IEEE Symposium on
  Security and Privacy (S $\&$ P)}, May 2019, pp. 691--706.

\bibitem{triastcyn2019federated}
A.~Triastcyn and B.~Faltings, ``{Federated Learning with Bayesian Differential
  Privacy},'' \emph{arXiv preprint arXiv:1911.10071}, 2019.

\bibitem{agarwal2018cpsgd}
N.~Agarwal, A.~T. Suresh, F.~X.~X. Yu, S.~Kumar, and B.~McMahan, ``cp{SGD}:
  Communication-efficient and differentially-private distributed {SGD},'' in
  \emph{Advances in Neural Information Processing Systems}, 2018, pp.
  7564--7575.

\bibitem{li2018federated}
T.~Li, A.~K. Sahu, M.~Zaheer, M.~Sanjabi, A.~Talwalkar, and V.~Smith,
  ``Federated optimization in heterogeneous networks,'' \emph{arXiv preprint
  arXiv:1812.06127}, 2018.

\bibitem{chen2018stochastic}
J.~Chen and R.~Luss, ``Stochastic gradient descent with biased but consistent
  gradient estimators,'' \emph{arXiv preprint arXiv:1807.11880}, 2018.

\bibitem{amiri2019machine}
M.~M. Amiri and D.~Gunduz, ``Machine learning at the wireless edge: Distributed
  stochastic gradient descent over-the-air,'' \emph{arXiv preprint
  arXiv:1901.00844}, 2019.

\bibitem{amiri2019federated}
------, ``Federated learning over wireless fading channels,'' \emph{arXiv
  preprint arXiv:1907.09769}, 2019.

\bibitem{timchang2020FL}
W.~T. {Chang} and R.~{Tandon}, ``Mac aware quantization for distributed
  gradient descent,'' in \emph{IEEE Global Communications Conference
  (GLOBECOM)}, 2020, pp. 1--6.

\bibitem{zhu2018low}
G.~{Zhu}, Y.~{Wang}, and K.~{Huang}, ``Broadband analog aggregation for
  low-latency federated edge learning,'' \emph{IEEE Transactions on Wireless
  Communications}, vol.~19, no.~1, pp. 491--506, 2020.

\bibitem{yang2020federated}
K.~Yang, T.~Jiang, Y.~Shi, and Z.~Ding, ``Federated learning via over-the-air
  computation,'' \emph{IEEE Transactions on Wireless Communications}, vol.~19,
  no.~3, pp. 2022--2035, 2020.

\bibitem{amiri2019over}
M.~M. Amiri and D.~G{\"u}nd{\"u}z, ``Over-the-air machine learning at the
  wireless edge,'' in \emph{2019 IEEE 20th International Workshop on Signal
  Processing Advances in Wireless Communications (SPAWC)}, July 2019, pp. 1--5.

\bibitem{zeng2019energy}
Q.~Zeng, Y.~Du, K.~K. Leung, and K.~Huang, ``Energy-efficient radio resource
  allocation for federated edge learning,'' \emph{arXiv preprint
  arXiv:1907.06040}, 2019.

\bibitem{sery2019analog}
T.~{Sery} and K.~{Cohen}, ``On analog gradient descent learning over multiple
  access fading channels,'' \emph{IEEE Transactions on Signal Processing},
  vol.~68, pp. 2897--2911, 2020.

\bibitem{wang2019adaptive}
S.~Wang, T.~Tuor, T.~Salonidis, K.~K. Leung, C.~Makaya, T.~He, and K.~Chan,
  ``Adaptive federated learning in resource constrained edge computing
  systems,'' \emph{IEEE Journal on Selected Areas in Communications}, vol.~37,
  no.~6, pp. 1205--1221, March 2019.

\bibitem{abad2019hierarchical}
M.~S.~H. Abad, E.~Ozfatura, D.~Gunduz, and O.~Ercetin, ``Hierarchical federated
  learning across heterogeneous cellular networks,'' \emph{arXiv preprint
  arXiv:1909.02362}, 2019.

\bibitem{khan2019federated}
L.~U. Khan, N.~H. Tran, S.~R. Pandey, W.~Saad, Z.~Han, M.~N. Nguyen, and C.~S.
  Hong, ``Federated learning for edge networks: Resource optimization and
  incentive mechanism,'' \emph{arXiv preprint arXiv:1911.05642}, 2019.

\bibitem{dwork2014algorithmic}
C.~Dwork, A.~Roth \emph{et~al.}, ``The algorithmic foundations of differential
  privacy,'' \emph{Foundations and Trends{\textregistered} in Theoretical
  Computer Science}, vol.~9, no. 3--4, pp. 211--407, 2014.

\bibitem{joseph2018local}
M.~Joseph, A.~Roth, J.~Ullman, and B.~Waggoner, ``Local differential privacy
  for evolving data,'' in \emph{Advances in Neural Information Processing
  Systems}, 2018, pp. 2375--2384.

\bibitem{geyer2017differentially}
R.~C. Geyer, T.~Klein, and M.~Nabi, ``Differentially private federated
  learning: A client level perspective,'' \emph{arXiv preprint
  arXiv:1712.07557}, 2017.

\bibitem{choudhury2019differential}
O.~Choudhury, A.~Gkoulalas-Divanis, T.~Salonidis, I.~Sylla, Y.~Park, G.~Hsu,
  and A.~Das, ``Differential privacy-enabled federated learning for sensitive
  health data,'' \emph{arXiv preprint arXiv:1910.02578}, 2019.

\bibitem{balle2018privacy}
B.~Balle, G.~Barthe, and M.~Gaboardi, ``Privacy amplification by subsampling:
  Tight analyses via couplings and divergences,'' \emph{Advances in Neural
  Information Processing Systems}, vol.~31, pp. 6277--6287, 2018.

\bibitem{balle2020privacy}
B.~Balle, P.~Kairouz, H.~B. McMahan, O.~Thakkar, and A.~Thakurta, ``Privacy
  amplification via random check-ins,'' \emph{arXiv preprint arXiv:2007.06605},
  2020.

\bibitem{seif2020wireless}
M.~{Seif}, R.~{Tandon}, and M.~{Li}, ``Wireless federated learning with local
  differential privacy,'' in \emph{IEEE International Symposium on Information
  Theory (ISIT)}, 2020, pp. 2604--2609.

\bibitem{liu2020privacy}
D.~Liu and O.~Simeone, ``Privacy for free: Wireless federated learning via
  uncoded transmission with adaptive power control,'' \emph{arXiv preprint
  arXiv:2006.05459}, 2020.

\bibitem{sonee2020efficient}
A.~Sonee and S.~Rini, ``Efficient federated learning over multiple access
  channel with differential privacy constraints,'' \emph{arXiv preprint
  arXiv:2005.07776}, 2020.

\bibitem{Dwork20061}
\BIBentryALTinterwordspacing
C.~Dwork, ``Differential privacy,'' in \emph{Automata, Languages and
  Programming: 33rd International Colloquium, ICALP 2006, Part II},
  M.~Bugliesi, B.~Preneel, V.~Sassone, and I.~Wegener, Eds., 2006, pp. 1--12.
  [Online]. Available: \url{https://doi.org/10.1007/11787006_1}
\BIBentrySTDinterwordspacing

\bibitem{smith2017interaction}
A.~Smith, A.~Thakurta, and J.~Upadhyay, ``Is interaction necessary for
  distributed private learning?'' in \emph{IEEE Symposium on Security and
  Privacy (S\&P)}, 2017, pp. 58--77.

\bibitem{erlingsson2020encode}
{\'U}.~Erlingsson, V.~Feldman, I.~Mironov, A.~Raghunathan, S.~Song, K.~Talwar,
  and A.~Thakurta, ``Encode, shuffle, analyze privacy revisited: formalizations
  and empirical evaluation,'' \emph{arXiv preprint arXiv:2001.03618}, 2020.

\bibitem{kairouz2015composition}
P.~Kairouz, S.~Oh, and P.~Viswanath, ``The composition theorem for differential
  privacy,'' in \emph{International conference on machine learning}.\hskip 1em
  plus 0.5em minus 0.4em\relax PMLR, 2015, pp. 1376--1385.

\bibitem{dwork2010boosting}
C.~Dwork, G.~N. Rothblum, and S.~Vadhan, ``Boosting and differential privacy,''
  in \emph{2010 IEEE 51st Annual Symposium on Foundations of Computer Science},
  October 2010, pp. 51--60.

\bibitem{ajalloeian2020analysis}
A.~Ajalloeian and S.~U. Stich, ``Analysis of sgd with biased gradient
  estimators,'' \emph{arXiv preprint arXiv:2008.00051}, 2020.

\bibitem{hasircioglu2020private}
B.~Hasircioglu and D.~Gunduz, ``Private wireless federated learning with
  anonymous over-the-air computation,'' \emph{arXiv preprint arXiv:2011.08579},
  2020.

\bibitem{tse2005fundamentals}
D.~Tse and P.~Viswanath, \emph{Fundamentals of wireless communication}.\hskip
  1em plus 0.5em minus 0.4em\relax Cambridge university press, 2005.

\bibitem{erlingsson2019amplification}
{\'U}.~Erlingsson, V.~Feldman, I.~Mironov, A.~Raghunathan, K.~Talwar, and
  A.~Thakurta, ``Amplification by shuffling: From local to central differential
  privacy via anonymity,'' in \emph{Proceedings of the Thirtieth Annual
  ACM-SIAM Symposium on Discrete Algorithms}.\hskip 1em plus 0.5em minus
  0.4em\relax SIAM, 2019, pp. 2468--2479.

\bibitem{rakhlin2012making}
A.~Rakhlin, O.~Shamir, and K.~Sridharan, ``Making gradient descent optimal for
  strongly convex stochastic optimization,'' in \emph{Proceedings of the 29th
  International Conference on International Conference on Machine
  Learning}.\hskip 1em plus 0.5em minus 0.4em\relax Omnipress, 2012, pp.
  1571--1578.

\end{thebibliography}

\appendices
\renewcommand{\thesectiondis}[2]{\Alph{section}:}
%Gaussian mechanism - basic definition
\section{Gaussian Mechanism for LDP\label{appendix:LDPdefinition}}
In this paper, we assume that each user's local perturbation noise is drawn from Gaussian distribution. This well-known technique is known as Gaussian mechanism and can provide rigorous privacy guarantees for LDP.

\begin{definition} (Gaussian Mechanism \cite{dwork2014algorithmic}) Suppose a user wants to release a function $f(X)$ of an input $X$ subject to $(\epsilon_{\ell}, \delta_{\ell})$-LDP. The Gaussian release mechanism is defined as $M(X) \triangleq f(X) + \mathcal{N}(0, \sigma^{2} \mathbf{I})$.
If the sensitivity of the function is bounded by $\Delta_f$, i.e., $\| f(x) - f(x') \|_{2}\leq \Delta_f$, $\forall x$, then for any $\delta_{\ell} \in (0,1]$, Gaussian mechanism satisfies $(\epsilon_{\ell}, \delta_{\ell})$-LDP, where $\epsilon_{\ell} = \frac{\Delta_{f}}{\sigma} \sqrt{2 \log \left(\frac{1.25}{\delta_{\ell}}\right)}$.
\end{definition} 

%Central privacy proof:
\section{Proof of Theorem \ref{theorem1_central_nonUniform}}\label{appendix:central_privacy_analysis}

In this section, we prove the privacy amplification due to non-uniform sampling of the users. For the per-iteration analysis, we drop the iteration index $t$ for brevity. Let $Y$ denote the output seen at the PS through MAC and $Y_{-k}$ denote the output when user $k$ does not participate. Recall that DP guarantees that any post-processing done on the received signal does not leak more information about the input. Therefore, it is sufficient to show the following,
\begin{align}
    \operatorname{Pr}(Y \in \mathcal{S}) \leq e^{\epsilon_{c}} \operatorname{Pr}(Y_{-k} \in \mathcal{S}) + \delta_c,~\forall k,\label{eq:genericEpsilonC}
\end{align}
and obtain $\epsilon_c$. The challenge of this proof is the random participation of users and that the local noises get aggregated over the wireless channel. In this case, let $\mathcal{K}$ denote the random set of users that participate in an iteration, and let $R=|\mathcal{K}|$ denote the random variable representing the number of participants. One can readily check that $R$ is a summation of $K$ Bernoulli random variables and has mean $\mu_R=\sum_{k=1}^{K} p_{k}$, where $p_{k}$ is the sampling probability of user $k$. The number of participants $R = |\mathcal{K}|$ determines the amplification of local DP via wireless aggregation, and in turn, determines the central DP. To take all possible $\mathcal{K}$ into account for the analysis, we condition the lefthand side of \eqref{eq:genericEpsilonC} with the event that $\mathcal{K}$ deviates from the mean, i.e., $|R - \mu_R | \geq \beta K$ for any $\beta > 0$, and bound it using Hoeffding's inequality and local DP guarantee. To apply local DP guarantee, we need additional conditioning on the event $\mathcal{E}_{k}$ that denotes the event where user $k$ participates in the training, i.e., $k \in \mathcal{K}$. Note that $p_k=\operatorname{Pr}(\mathcal{E}_{k}),\forall k$ and the conditional probabilities $\bar{p}_k=\operatorname{Pr}(\mathcal{E}_{k}||R - \mu_R | \geq \beta K),\forall k$ can be readily bounded by $p_k$'s using total probability theorem and Hoeffding's inequality, i.e., one can show that $\bar{p}_k\leq p_k/(1-\delta')$. For any $k\in [K]$, we have the following inequalities:
\begin{align}
    \operatorname{Pr}(Y \in \mathcal{S}) & = \begin{aligned}[t]&\operatorname{Pr}(|R - \mu_R | \geq \beta K) \operatorname{Pr}(Y \in \mathcal{S} | |R - \mu_R | \geq \beta K)\\ +& \operatorname{Pr}(|R - \mu_R | < \beta K) \operatorname{Pr}(Y \in \mathcal{S} | |R - \mu_R | < \beta K) \end{aligned}\nonumber \\ 
    & \leq \delta' \times 1 + \operatorname{Pr}(|R - \mu_R | < \beta K) \operatorname{Pr}(Y \in \mathcal{S} | |R - \mu_R | < \beta K), \label{eq:centralProofStep}
\end{align}
where the inequality follows from the fact that any probability is upper bounded by $1$ and from the Lemma below:
\begin{lemma}
(Hoeffding's Inequality for Binomial Random Variable) For a binomial random variable $X$ with $K$ trials and mean $\mu_X$, the probability that $X$ deviates from the mean by more than $\beta K$ can be bounded as,
\begin{align}
    \operatorname{Pr} (|X - \mu_{X}| \geq \beta K) \leq 2 e^{- 2 \beta^{2} K} \triangleq \delta', 
\end{align}
for any $\beta >0$, and any $\delta'\in [0,1)$.
\label{lemma:Hoeffding}
\end{lemma}
To further upper bound \eqref{eq:centralProofStep}, we use the following Lemma.
\begin{lemma}
Let $\bar{p}_k=\operatorname{Pr}( \mathcal{E}_{k} | |R - \mu_R | < \beta K)$ and $c$ be some constant that depends on the privacy mechanism, specifically for the Gaussian mechanism we have $c \triangleq  \frac{ 2  L}{\sigma_{\min}}  \sqrt{2 \log \frac{1.25}{\delta_{\ell}} }$, where $L$ is the Lipschitz constant. The following inequality is true when the local mechanism satisfies $(c/\sqrt{\mu_R-\beta K}, \delta_{\ell})$-LDP:
\begin{align}
    \operatorname{Pr}(Y \in \mathcal{S} | |R - \mu_R | < \beta K)  \leq \left[ \bar{p}_k \left( e^{\frac{c}{\sqrt{\mu_R - \beta K}}} -1 \right) + 1 \right]  \operatorname{Pr}(Y_{-k} \in \mathcal{S} | |R - \mu_R | < \beta K) + \bar{p}_k \delta_{\ell}.\nonumber
\end{align}
\label{lemma:conditonalBoundedByDP}
\end{lemma}
\vspace{-0.6in}
Using Lemma \ref{lemma:conditonalBoundedByDP}, we can bound \eqref{eq:centralProofStep} as follows:
\begin{align}
     &\operatorname{Pr}(Y \in \mathcal{S}) \nonumber\\
     & \leq \delta' + {\operatorname{Pr}(|R - \mu_R | < \beta K) } \left[ \left[ \bar{p}_k \left( e^{\frac{c}{\sqrt{\mu_R - \beta K}}} -1 \right) + 1 \right]  \operatorname{Pr}(Y_{-k} \in \mathcal{S} | |R - \mu_R | < \beta K) + \bar{p}_k \delta_{\ell} \right] \nonumber \\ 
    & \overset{(a)} \leq \delta' + \bar{p}_k \delta_{\ell} + {\operatorname{Pr}(|R -\mu_R | < \beta K) }  \left[ \bar{p}_k \left( e^{\frac{c}{\sqrt{\mu_R - \beta K}}} -1 \right) + 1 \right] \frac{\operatorname{Pr}(Y_{-k} \in \mathcal{S} )}{\operatorname{Pr}(|R -\mu_R | < \beta K) }   \nonumber \\
    & \overset{(b)} \leq \delta' + \frac{p_k}{1-\delta'} \delta_{\ell} +  \left[ \frac{p_k}{1-\delta'} \left( e^{\frac{c}{\sqrt{\mu_R - \beta K}}} -1 \right) + 1 \right] \operatorname{Pr}(Y_{-k} \in \mathcal{S})
\end{align}
where $(a)$ follows from total probability theorem and the fact that ${\operatorname{Pr}(|R - \mu_R | < \beta K) } \bar{p}_k\delta_{\ell} \leq \bar{p}_k\delta_{\ell}$; and $(b)$ follows from inequality $\bar{p}_k\leq p_k/(1-\delta')$ mentioned at the beginning of the proof. We can obtain a bound for each user $k$ in a similar fashion. By selecting the bound that gives us the largest privacy parameters, we recover the result of Theorem \ref{theorem1_central_nonUniform}.
% Since the bound needs to be satisfied by all users, we select the one that gives us the largest privacy parameters. Therefore, we have 
% \begin{align}
%     \epsilon_{c} &= \log \left[ \frac{\max_k p_k}{1-\delta'} \left( e^{\frac{c}{\sqrt{\mu_R - \beta K}}} -1 \right) + 1 \right], \quad
%     \delta_{c} = \delta' + \frac{\max_k p_k \delta_{\ell}}{1-\delta'}.
% \end{align}
% This concludes the proof of Theorem \ref{theorem1_central_nonUniform}. 
We next prove Lemma \ref{lemma:conditonalBoundedByDP}.

\begin{proof}[Proof of Lemma \ref{lemma:conditonalBoundedByDP}]

With the $\mathcal{E}_{k}$ defined above,
% denote the event where user $k$ is sampled, i.e., $k \in \mathcal{K}$, and
let $\mathcal{E}_{k}^{c}$ denote its complementary event. Then, using total probability theorem, we have 
\begin{align}
    &\operatorname{Pr}(Y \in \mathcal{S} | |R - \mu_R | < \gamma)\nonumber\\
    % & = \operatorname{Pr}( \mathcal{E}_{k}  | |R - \mu_R | < \gamma)  \operatorname{Pr}(Y \in \mathcal{S} | |R - \mu_R | < \gamma,   \mathcal{E}_{k} ) + \operatorname{Pr}( \mathcal{E}_{k}^{c}  | |R - \mu_R | < \gamma)  \operatorname{Pr}(Y \in \mathcal{S} | |R - \mu_R | < \gamma,   \mathcal{E}_{k}^{c} ) \nonumber \\ 
    & = \bar{p}_k \operatorname{Pr}(Y \in \mathcal{S} | |R - \mu_R | < \gamma,   \mathcal{E}_{k} ) + (1-\bar{p}_k)  \operatorname{Pr}(Y \in \mathcal{S} | |R - \mu_R | < \gamma,   \mathcal{E}_{k}^{c} ) \nonumber\\
     & \overset{(a)} = \bar{p}_k \operatorname{Pr}(Y \in \mathcal{S} | |R - \mu_R | < \gamma,   \mathcal{E}_{k} ) + (1-\bar{p}_k)  \operatorname{Pr}(Y_{-k} \in \mathcal{S} | |R - \mu_R | < \gamma),\label{eq:ProbYinSConditionGamma}
\end{align}
where we can show that $(a)$ is true as follows,
\begin{align}
    & \operatorname{Pr}(Y \in \mathcal{S} | |R - \mu_R | < \gamma,   \mathcal{E}_{k}^{c} )  \nonumber \\ 
       & = \sum_{\substack{A_{-k} \subseteq [K],\\ ||A_{-k}| - \mu_R | < \gamma }} \operatorname{Pr}(\mathcal{K} =A_{-k} | |R - \mu_R| < \gamma, \mathcal{E}_{k}^{c})  \operatorname{Pr} (Y \in \mathcal{S} | |R - \mu_R| < \gamma, \mathcal{E}_{k}^{c}, \mathcal{K} = A_{-k}) \nonumber \\ 
        & \overset{(a)} =  \sum_{\substack{A_{-k} \subseteq [K],\\ ||A_{-k}| - \mu_R | < \gamma }} \operatorname{Pr}(\mathcal{K} =A_{-k} | |R - \mu_R| < \gamma)  \operatorname{Pr} (Y \in \mathcal{S} | |R - \mu_R| < \gamma, \mathcal{E}_{k}^{c}, \mathcal{K} = A_{-k}) \nonumber \\ 
    & \overset{(b)} = \sum_{\substack{A_{-k} \subseteq [K],\\ ||A_{-k}| - \mu_R | < \gamma }} \operatorname{Pr}(\mathcal{K} =A_{-k} | |R - \mu_R| < \gamma)  \operatorname{Pr} (Y_{-k} \in \mathcal{S} | |R - \mu_R| < \gamma, \mathcal{K} = A_{-k}) \nonumber \\ 
      & = \operatorname{Pr} (Y_{-k} \in \mathcal{S} | |R - \mu_R| < \gamma)\label{eq:ProbConditionEkC}
\end{align}
where $(a)$ holds since user $k$ is not in the set $A_{-k}$, therefore, conditioning on the event $\mathcal{E}_{k}^{c}$ does not change the probability; and $(b)$ follows due to similar argument. Next, we upper bound $\operatorname{Pr}(Y \in \mathcal{S} | |R - \mu_R | < \gamma,   \mathcal{E}_{k} )$ as follows:
\begin{align}
    &\operatorname{Pr}(Y \in \mathcal{S} | |R - \mu_R | < \gamma,   \mathcal{E}_{k} )\nonumber\\
    & = \sum_{\substack{A \subseteq [K]: k \in A,\\ ||A| - \mu_R | < \gamma }} \operatorname{Pr}(\mathcal{K} =A  | |R - \mu_R| < \gamma, \mathcal{E}_{k}) \operatorname{Pr} (Y \in \mathcal{S} | |R - \mu_R| < \gamma, \mathcal{E}_{k}, \mathcal{K} = A), \label{conditionalYsample}
\end{align}
Note that, in wireless setting, when each user $k$ applies a mechanism that satisfies $(\epsilon_{\ell}, \delta_{\ell})$-LDP, it implies $(c/\sqrt{|A|}, \delta_{\ell})$-DP \cite{seif2020wireless} (using quasi-convexity property of DP \cite{erlingsson2019amplification}), we have,
\begin{align}
     \hspace{-5pt}\operatorname{Pr} (Y \in \mathcal{S} | |R - \mu_R| < \gamma, \mathcal{E}_{k}, \mathcal{K} = A) \leq e^{\frac{c}{\sqrt{|A|}}} \operatorname{Pr} (Y \in \mathcal{S} | |R - \mu_R| < \gamma, \mathcal{E}_{k}^{c}, \mathcal{K} = A_{-k}) + \delta_{\ell}. \label{localPrivacyLevel}
\end{align}
Plugging \eqref{localPrivacyLevel} into \eqref{conditionalYsample}, we obtain the following:
\begin{align}
      & \operatorname{Pr}(Y \in \mathcal{S} | |R - \mu_R | < \gamma,   \mathcal{E}_{k} )  \nonumber \\ 
       & \leq  \sum_{\substack{A \subseteq [K]: k \in A,\\ ||A| - \mu_R | < \gamma }} \operatorname{Pr}(\mathcal{K} =A | |R - \mu_R| < \gamma, \mathcal{E}_{k}) \left[e^{\frac{c}{\sqrt{|A|}}} \operatorname{Pr} (Y \in \mathcal{S} | |R - \mu_R| < \gamma, \mathcal{E}_{k}^{c}, \mathcal{K} = A_{-k}) + \delta_{\ell} \right] \nonumber \\ 
        & \overset{(a)} =  \sum_{\substack{A \subseteq [K]: k \in A,\\ ||A| - \mu_R | < \gamma }} \operatorname{Pr}(\mathcal{K} = A_{-k} | |R - \mu_R| < \gamma) e^{\frac{c}{\sqrt{|A|}}} \operatorname{Pr} (Y \in \mathcal{S} | |R - \mu_R| < \gamma, \mathcal{E}_{k}^{c}, \mathcal{K} = A_{-k}) + \delta_{\ell} \nonumber \\ 
        & \overset{(b)} =  \sum_{\substack{A \subseteq [K]: k \in A,\\ ||A| - \mu_R | < \gamma }} \operatorname{Pr}(\mathcal{K} = A_{-k} | |R - \mu_R| < \gamma) e^{\frac{c}{\sqrt{|A|}}} \operatorname{Pr} (Y_{-k} \in \mathcal{S} | |R - \mu_R| < \gamma, \mathcal{K} = A_{-k}) + \delta_{\ell} \nonumber \\ 
    & \overset{(c)}     \leq e^{\frac{c}{\sqrt{\mu_R -\gamma}}}  \sum_{\substack{A \subseteq [K]: k \in A,\\ ||A| - \mu_R | < \gamma }} \operatorname{Pr}(\mathcal{K} = A_{-k} | |R - \mu_R| < \gamma)  \operatorname{Pr} (Y_{-k} \in \mathcal{S} | |R - \mu_R| < \gamma, \mathcal{K} = A_{-k}) + \delta_{\ell} \nonumber \\ 
      & = e^{\frac{c}{\sqrt{\mu_R -\gamma}}}   \operatorname{Pr} (Y_{-k} \in \mathcal{S} | |R - \mu_R| < \gamma) + \delta_{\ell} \label{eq:ProbConditionEk}
\end{align}
where $(a)$ 
% holds due to the fact that $\operatorname{Pr}(\mathcal{K} =A | |R - \mu_R| < \gamma, \mathcal{E}_{k})$ is equivalent to $\operatorname{Pr}(\mathcal{K} =A_{-k} | |R - \mu_R| < \gamma)$ when we condition on the event $\mathcal{E}_{k}$; 
and $(b)$ follows the similar argument as the one used in \eqref{eq:ProbConditionEkC}. From the condition on the cardinality of the set $R$, we know that $R=|A|$ and $\mu_R - \gamma < |A| < \mu_R + \gamma$. Therefore, $(c)$ follows from using the lower bound on $|A|$.
% we can use the lower bound on $|A|$ to upper bound the exponential term, then $(c)$ follows.
Then, by combining \eqref{eq:ProbYinSConditionGamma}, \eqref{eq:ProbConditionEkC} and \eqref{eq:ProbConditionEk}, we have
\begin{align}
    &\operatorname{Pr}(Y \in \mathcal{S} | |R - \mu_R | < \gamma) \nonumber\\
    \leq & ~\bar{p}_k e^{\frac{c}{\sqrt{\mu_R -\gamma}}}   \operatorname{Pr} (Y_{-k} \in \mathcal{S} | |R - \mu_R| < \gamma) + \bar{p}_k \delta_{\ell} + (1-\bar{p}_k)  \operatorname{Pr}(Y_{-k} \in \mathcal{S} | |R - \mu_R | < \gamma).
\end{align}
Rearranging the above inequality, we recover the result of Lemma \ref{lemma:conditonalBoundedByDP}.
\end{proof}

%Optimal sampling probability 
\section{Proof of Lemma \ref{lemma:optimalSamplingProb}} \label{optimal_sampling_probability}
In this section, we find the optimal sampling probability $p_t^*$ that minimizes the central privacy level $\epsilon_{c,t}$ for the wireless FedSGD scheme. For the per-iteration analysis, we drop the iteration index for brevity. We minimize $\epsilon_c$ as follows:
\begin{align}
        \epsilon_{c} &= \log \left[ \frac{p}{1-\delta'} \left( e^{\frac{c}{\sqrt{Kp - \beta K}}} -1 \right) + 1 \right] \label{global_privacy}  \leq  \frac{p}{1-\delta'} \left( e^{\frac{c}{\sqrt{Kp - \beta K}}} -1 \right).
\end{align}
We assume that $p$ takes the form of $\tilde{k}/K$, i.e., $p = \frac{\tilde{k}}{K}$, then 
\begin{align}
   \epsilon_{c} \leq \frac{\tilde{k}}{K(1-\delta')} \left( e^{\frac{c}{\sqrt{\tilde{k} - \beta K}}} -1 \right) \triangleq \tilde{\epsilon}_{c}.
\end{align}
Taking the derivative of the right-hand side w.r.t. $\tilde{k}$ and setting it to zero yields the following: 
\begin{align}
    \frac{\partial \tilde{\epsilon}_{c}}{\partial \tilde{k}} = \frac{c}{K (1-\delta')} \left[ - \frac{\tilde{k}}{2} (\tilde{k} - \beta K)^{\frac{-3}{2}}  + (\tilde{k} - \beta K)^{\frac{-1}{2}}\right] = 0 \Rightarrow \tilde{k} = 2 \beta K.
\end{align}
We then check the second derivative of the right-hand side and obtain, 
\begin{align}
        \frac{\partial^{2} \tilde{\epsilon}_{c}}{\partial \tilde{k}^{2}} & \propto 
        % -\frac{1}{2} \left[\tilde{k} \times -\frac{3}{2} (\tilde{k} -\beta K)^{-5/2} + (\tilde{k} -\beta K)^{-3/2} \right] - \frac{1}{2} (\tilde{k} -\beta K)^{-3/2} \nonumber \\ 
        % & = 
        \frac{-1}{2} \times \left[- \frac{3}{2} \tilde{k} (\tilde{k} - \beta K)^{-5/2} + 2 (\tilde{k} - \beta K)^{-3/2} \right].
\end{align}
It can be readily shown that  $   \frac{\partial^{2} \tilde{\epsilon}_{c}}{\partial \tilde{k}^{2}} \leq 0$ when $\tilde{k} \geq 2 \beta K$. To this end, the optimal sampling probability that minimize $\epsilon_{c}$ is $p^{*} = 2\beta$.
Using Lemma \ref{lemma:Hoeffding}, 
we know that $\beta = \frac{1}{\sqrt{K}} \sqrt{\frac{1}{2} \log \left(\frac{2}{\delta'}\right)}$.
By plugging $p^*$ and $\beta$ into \eqref{global_privacy},
% and defining $\alpha=\sqrt{\frac{1}{2}\log\left(\frac{2}{\delta'}\right)}$, 
we get:
\begin{align}
 \epsilon_{c} = \log \left[ \frac{2\sqrt{\frac{1}{2}\log\left(\frac{2}{\delta'}\right)}}{\sqrt{K}(1-\delta')} \big(e^{\frac{c}{\sqrt[4]{K\frac{1}{2}\log\left(\frac{2}{\delta'}\right)}}} -1 \big) + 1\right] = \mathcal{O}\left(\frac{1}{K^{3/4}}\right).
\end{align}
This completes the proof of Lemma \ref{lemma:optimalSamplingProb}.

 %Local privacy:

 \section{Proof of Lemma \ref{lemma1:local_privacy}} \label{appendix::sensitivity_analysis}
The final received signal at the PS from \eqref{eq:output} can be expressed as:
$    \mathbf{y}_{t} =  \sum_{k \in \mathcal{K}_{t}} h_{k,t} \alpha_{k,t} {\mathbf{g}_{k}(\mathbf{w}_{t})} + \mathbf{z}_t$ and the variance of the effective Gaussian noise $ \mathbf{z}_t$ is 
$\sigma^{2} =  \sum_{k \in \mathcal{K}_{t}} h_{k,t}^{2} \alpha_{k,t}^{2} \sigma_{k,t}^{2}  + N_{0}$. In order to invoke the result of the Gaussian mechanism (Appendix \ref{appendix:LDPdefinition}), we next obtain a bound on the sensitivity for user $k$. To bound the local sensitivity of user $k$, we fix the gradients of the remaining ${\mathcal{K}_{t}} \backslash k$ users. The local sensitivity of user $k$ can then be bounded as
\begin{align}
    \Delta_{k,t} 
    & = \max_{\mathcal{D}_{k}, \mathcal{D}'_{k}} 
||  \mathbf{y}_t - \mathbf{y}^{'}_t ||_2 = \max_{\mathcal{D}_{k}, \mathcal{D}'_{k}} 
||  h_{k,t} \alpha_{k,t} (\mathbf{g}_{k}(\mathbf{w}_{t})- \mathbf{g'}_{k}(\mathbf{w}_{t}))||_2 \nonumber\\
& \leq  h_{k,t} \alpha_{k,t}  \max_{\mathcal{D}_{k}, \mathcal{D}'_{k}} 
||  \mathbf{g}_{k}(\mathbf{w}_{t})||_2 + ||  \mathbf{g'}_{k}(\mathbf{w}_{t})||_2 \overset{(a)}\leq 2 h_{k,t} \alpha_{k,t}  L \label{eq:sensitivity} \overset{(b)} = 2 L,
\end{align}
where $(a)$ follows from the fact that $\| {\mathbf{g}_{k}(\mathbf{w}_{t})} \|_{2}   \leq L, \forall k$; and $(b)$ follows from the channel inversion transmission scheme. 
We next show the guarantee on the local DP of user $k$ when user $k$ is a participant. Following similar steps used for proving \eqref{eq:centralProofStep}, it can be shown that,
\begin{align}
       \operatorname{Pr}(Y_x^{(k)} \in \mathcal{S}| \mathcal{E}_{k})  & \leq   \delta' + \delta_{\ell} +  e^{\frac{c}{\sqrt{1 + \mu_R -\beta K}}} \operatorname{Pr}(Y_{x'}^{(k)} \in \mathcal{S}| \mathcal{E}_{k}),\label{eq:localPrivacyConditionK}
\end{align}
where $\mu_{R} = \sum_{i=1, i \neq k}^{K} p_{i}$. Note that \eqref{eq:localPrivacyConditionK} is conditioning on the event when user $k$ participates. We next use the total probability theorem and obtain the following set of steps:
\begin{align}
    \operatorname{Pr}(Y_{x}^{(k)} \in \mathcal{S} ) & = p_{k}   \operatorname{Pr}(Y_{x}^{(k)} \in \mathcal{S} | \mathcal{E}_{k}) + (1-p_{k})  \operatorname{Pr}(Y_{x}^{(k)} \in \mathcal{S} | \mathcal{E}_{k}^{c}) \nonumber \\ 
    % &\overset{(a)}{\leq} p_{k}  e^{\epsilon_{\ell}} \operatorname{Pr}(Y_{x'}^{(k)} \in \mathcal{S} | \mathcal{E}_{k}) + p_{k} (\delta_{\ell} + \delta') + (1-p_{k}) e^{0} \operatorname{Pr}(Y_{x'}^{(k)} \in \mathcal{S} | \mathcal{E}_{k}^{c}) \nonumber \\
    & \overset{(a)}{\leq}   p_{k}  e^{\epsilon_{\ell}}  \operatorname{Pr}(Y_{x'}^{(k)} \in \mathcal{S} | \mathcal{E}_{k}) + p_{k} (\delta_{\ell} + \delta') + (1-p_{k}) e^{\epsilon_{\ell}}  \operatorname{Pr}(Y_{x'}^{(k)} \in \mathcal{S} | \mathcal{E}_{k}^{c}) \nonumber\\
    & =    e^{\epsilon_{\ell}}  \operatorname{Pr}(Y_{x'}^{(k)} \in \mathcal{S} ) + p_{k} (\delta_{\ell} + \delta'),
\end{align}
where step $(a)$ follows from \eqref{eq:localPrivacyConditionK} and the fact that when user $k$ is not participating, we have 
\begin{align}
   \operatorname{Pr}(Y_{x}^{(k)} \in \mathcal{S} | \mathcal{E}_{k}^{c}) = e^{0} \operatorname{Pr}(Y_{x'}^{(k)} \in \mathcal{S} | \mathcal{E}_{k}^{c}) \leq e^{\epsilon_{\ell}} \operatorname{Pr}(Y_{x'}^{(k)} \in \mathcal{S} | \mathcal{E}_{k}^{c}), \forall x, x'. 
\end{align}
We arrive at the proof of Lemma \ref{lemma1:local_privacy}.

%Convergence proofs: 
\section{Proofs of Theorem  \ref{theorem_1_convergence} and Theorem \ref{theorem_1_convergence_known_kt} }\label{appendix:proof_theorem_2}
\vspace{-2pt}
%\subsection{With i.i.d. data}
When the data is i.i.d., we can invoke a slightly modified version of the result of \cite{rakhlin2012making} on convergence of SGD for $\mu$-smooth and $\lambda$-strongly convex loss, which states 
% \begin{align}
%      \mathds{E} \left[ F(\mathbf{w}_{T}) \right] - F(\mathbf{w}^{*}) \leq \frac{2 \mu G^{2}}{\lambda^{2} T}.
% \end{align}
% By slight modification of the result in \cite{rakhlin2012making}, we get 
\begin{align}
     \mathds{E} \left[ F(\mathbf{w}_{T}) \right] - F(\mathbf{w}^{*}) \leq \frac{2 \mu}{\lambda^{2} T} \left(\sum_{t=1}^{T} G_{t}^{2}/T \right),\label{eq:finalconv2}
\end{align}
where $ G_{t}^{2}$ is the upper bound on the second moment of the gradient estimate, i.e., $\mathds{E}\left[\|\hat{\mathbf{g}}_{t}\|_{2}^{2}\right] \leq G_{t}^{2}$.

%Case1:
\subsection{$\mathcal{K}_t$ is Unknown at the PS}
To prove  the convergence rate of the proposed algorithm, we recall that the gradient estimate at the PS in (\ref{eq:postpreprocessing}) needs to satisfy: (a) Unbiasedness, i.e., $\mathds{E} \left[\hat{\mathbf{g}}_{t} \right] = \mathbf{g}_{t}$, since the total additive noise is zero mean; and (b) Bounded second moment,  $\mathds{E}\left[\|\hat{\mathbf{g}}_{t}\|_{2}^{2}\right] \leq G_{t}^{2}$, which we prove as follows. 
% Let us denote $\mu_{|\mathcal{K}_{t}|} = \sum_{k = 1}^{K} p_{k,t}$ as the average number of sampled users at each iteration. Also, denote $\sigma_{|\mathcal{K}_{t}|}^{2} = \sum_{k=1}^{K} p_{k,t}(1-p_{k,t})$. 
Recall that the estimated gradient at the PS is 
\begin{align}
    \hat{\mathbf{g}}_{t} & = {  \frac{1}{\mu_{|\mathcal{K}_{t}|} }  \sum_{k \in \mathcal{K}_{t}}  {\mathbf{g}_{k}(\mathbf{w}_{t})}} + \frac{1}{ \mu_{|\mathcal{K}_{t}|} }    \mathbf{z}_{t}  =  \frac{1}{\mu_{|\mathcal{K}_{t}|}}   \sum_{k \in \mathcal{K}_{t}} \frac{1}{b}\sum_{i\in\mathcal{B}_k}  \nabla f_{k}(\mathbf{w}_t; (\mathbf{u}_{i}^{(k)}, v_{i}^{(k)}) ) + \frac{1}{ \mu_{|\mathcal{K}_{t}|} }   \mathbf{z}_{t}. \nonumber 
\end{align}
By taking the expectation over the randomness of SGD, user sampling and noise, we have
\begin{align}
    \mathds{E} \left[\hat{\mathbf{g}}_{t} \right] & = \frac{1}{\mu_{|\mathcal{K}_{t}|} b}  \mathds{E} \left[ \sum_{k \in \mathcal{K}_{t}}  \sum_{i\in\mathcal{B}_k} \nabla f_{k}(\mathbf{w}_t; (\mathbf{u}_{i}^{(k)}, v_{i}^{(k)}) ) \right] = \frac{1}{\mu_{|\mathcal{K}_{t}|} b} \mathds{E} \left[ |\mathcal{K}_{t}| \right] b\mathbf{g}_{t} = \mathbf{g}_{t}.
\end{align}
Therefore, the estimated gradient is unbiased. We next obtain the bound on the second moment of the estimated gradient. We have 
\begin{align}
     \mathds{E}\left[\|\hat{\mathbf{g}}_{t}\|_{2}^{2}\right] & = \mathds{E} \left[\Vert \frac{1}{\mu_{|\mathcal{K}_{t}|} } \sum_{k \in \mathcal{K}_{t}} \mathbf{g}_{k}(\mathbf{w}_{t}) +  \frac{\mathbf{z}_{t}}{\mu_{|\mathcal{K}_{t}|}  } \Vert^{2}_{2} \right] \nonumber\\
        & \overset{(a)}= \frac{1}{\mu_{|\mathcal{K}_{t}|}^{2}} \mathds{E} \left[\| \sum_{k \in \mathcal{K}_{t}} \mathbf{g}_{k}(\mathbf{w}_{t})\|^{2}_{2}\right] + \frac{1}{ \mu_{|\mathcal{K}_{t}|}^{2} } \mathds{E}\left[ \| \mathbf{z}_{t}\|^{2}_{2} \right] \nonumber\\
      & = \frac{1}{\mu_{|\mathcal{K}_{t}|}^{2}}  \mathds{E} \left[\sum_{k \in \mathcal{K}_{t}} \|  \mathbf{g}_{k}(\mathbf{w}_{t})\|^{2}_{2} + \sum_{k \in \mathcal{K}_{t}} \sum_{k' \in \mathcal{K}_{t}} \mathbf{g}_{k}(\mathbf{w}_{t})^{T} \mathbf{g}_{k'}(\mathbf{w}_{t})  \right] +  \frac{1}{ \mu_{|\mathcal{K}_{t}|}^{2} } \mathds{E}\left[ \| \mathbf{z}_{t}\|^{2}_{2} \right] \nonumber \\
       & \overset{(b)} \leq  \frac{1}{\mu_{|\mathcal{K}_{t}|}^{2}}  \mathds{E} \left[\sum_{k \in \mathcal{K}_{t}} \|  \mathbf{g}_{k}(\mathbf{w}_{t})\|^{2}_{2} + \sum_{k \in \mathcal{K}_{t}} \sum_{k' \in \mathcal{K}_{t}} \| \mathbf{g}_{k}(\mathbf{w}_{t})\|_{2} \|\mathbf{g}_{k'}(\mathbf{w}_{t})\|_{2}  \right] +  \frac{1}{ \mu_{|\mathcal{K}_{t}|}^{2} } \mathds{E}\left[ \| \mathbf{z}_{t}\|^{2}_{2} \right] \nonumber \\ 
          &\overset{(c)} \leq  \frac{1}{\mu_{|\mathcal{K}_{t}|}^{2}}  \mathds{E} \left[|\mathcal{K}_{t}| L^{2} + |\mathcal{K}_{t}| (|\mathcal{K}_{t}| - 1) L^{2} \right] +  \frac{1}{ \mu_{|\mathcal{K}_{t}|}^{2} } \mathds{E}\left[ \| \mathbf{z}_{t}\|^{2}_{2} \right]  =   \frac{L^{2} \mathds{E} \left[|\mathcal{K}_{t}|^{2} \right]}{\mu_{|\mathcal{K}_{t}|}^{2}} + \frac{1}{ \mu_{|\mathcal{K}_{t}|}^{2}} \mathds{E}\left[ \| \mathbf{z}_{t}\|^{2}_{2} \right] \nonumber\\
          & \leq   \frac{L^{2} \times (\mu_{|\mathcal{K}_{t}|}^{2} + \sigma_{|\mathcal{K}_{t}|}^{2}) }{\mu_{|\mathcal{K}_{t}|}^{2}} + \frac{d}{ \mu_{|\mathcal{K}_{t}|}^{2} } \times \left[   \max_{k} \sigma_{k,t}^{2} \times  \mathds{E}\left[ |\mathcal{K}_{t}|\right] +N_{0}\right]\nonumber\\
            & =  \frac{L^{2} \times (\mu_{|\mathcal{K}_{t}|}^{2} + \sigma_{|\mathcal{K}_{t}|}^{2}) }{\mu_{|\mathcal{K}_{t}|}^{2}} + \frac{d}{\mu_{|\mathcal{K}_{t}|}^{2} } \times \left[   \max_{k} \sigma_{k,t}^{2} \times  \mu_{|\mathcal{K}_{t}|} +N_{0}\right] \triangleq G_{t}^{2}, \label{eq:finalconv}
\end{align}
where (a) follows from the fact that ${ \mathds{E}\left[\mathbf{g}_{t}^{T} \mathbf{z}_{t}\right] = 0 } $,  (b) follows from Cauchy-Schwarz inequality, and (c) from the assumption that $ \| \mathbf{g}_{k}(\mathbf{w}_{t})\|_{2} \leq L$, i.e., the Lipschitz constant $\forall k$. 
Plugging $G_{t}^{2}$ from \eqref{eq:finalconv} in \eqref{eq:finalconv2}, we arrive at the proof of Theorem \ref{theorem_1_convergence}.

%Case2:

\subsection{$\mathcal{K}_t$ is Known at the PS}
We then move to the case when $\mathcal{K}_t$ is known at the PS. Recall that the estimated gradient at the PS for the known $\mathcal{K}_t$ case is
\begin{align}
    \hat{\mathbf{g}}_{t}  = {  \frac{1}{\zeta_t |\mathcal{K}_{t}| } \sum_{k \in \mathcal{K}_{t}}  {\mathbf{g}_{k}(\mathbf{w}_{t})}} + \frac{1}{\zeta_t  |\mathcal{K}_{t}| }  \mathbf{z}_{t}  =  \frac{1}{\zeta_t |\mathcal{K}_{t}|} \sum_{k \in \mathcal{K}_{t}} \frac{1}{b} \sum_{i\in\mathcal{B}_k}  \nabla f_{k}(\mathbf{w}_t; (\mathbf{u}_{i}^{(k)}, v_{i}^{(k)}) ) + \frac{1}{ \zeta_t |\mathcal{K}_{t}| } \mathbf{z}_{t},
\end{align}
where $\zeta_t$ is used for maintaining unbiasedness of the estimated gradient and will be specified later. By taking the expectation over the randomness of SGD, user sampling and additive noise, we have
\begin{align}
    \mathds{E} \left[\hat{\mathbf{g}}_{t} \right] & = \mathds{E} \left[ \frac{1}{\zeta_t |\mathcal{K}_{t}|}  \sum_{k \in \mathcal{K}_{t}} \frac{1}{b} \sum_{i\in\mathcal{B}_k} \nabla f_{k}(\mathbf{w}_t; (\mathbf{u}_{i}^{(k)}, v_{i}^{(k)}) ) \right] \nonumber\\
    & = \mathds{E} \left[ \frac{1}{\zeta_t |\mathcal{K}_{t}|}  |\mathcal{K}_{t}|  \mathbf{g}_{t}\right] = \frac{1}{\zeta_t}\big(1- \prod_{k=1}^{K} (1 - p_{k,t})\big) \mathbf{g}_{t} = \mathbf{g}_{t}.
\end{align}
In order get unbiased estimate for $\mathbf{g}_{t}$, $\zeta_t$ is chosen as $\zeta_{t} = 1- \prod_{k=1}^{K} (1 - p_{k,t}) $. To bound the second moment, the proof follows similar steps as the unknown $\mathcal{K}_t$ case, and is omitted due to space limitation.

\end{document}